\documentclass[11pt]{article}

\usepackage[sort&compress,round]{natbib}

\usepackage{amsmath}
\usepackage{amsthm}
\usepackage{amssymb}

\usepackage{graphicx}
\setkeys{Gin}{width=\linewidth,totalheight=\textheight,keepaspectratio}
\graphicspath{{./images/}}

\usepackage{booktabs}

\usepackage{units}

\usepackage{fancyvrb}
\fvset{fontsize=\normalsize}

\usepackage[T1]{fontenc}
\usepackage[sc]{mathpazo}
\usepackage{multicol}
\usepackage{multirow}
\usepackage{lipsum}
\usepackage{wrapfig}
\usepackage{latexsym}
\usepackage{mathrsfs}
\usepackage{mathtools}
\usepackage{enumerate}
\usepackage{chemarrow}
\usepackage[shortlabels]{enumitem}
\setlist{topsep=0.5em, itemsep=0em}
\usepackage{tensor}

\usepackage[math]{cellspace}
\usepackage{hyperref}
\hypersetup{
	hidelinks = true
}
\usepackage{cleveref}
\usepackage{color,cancel}
\usepackage[normalem]{ulem}
\usepackage{verbatim}
\usepackage{tcolorbox}
\usepackage{soul}

\newtheorem{theorem}{Theorem}[section]
\newtheorem{proposition}[theorem]{Proposition}
\newtheorem{definition}[theorem]{Definition}
\newtheorem{corollary}[theorem]{Corollary}

\newcounter{assume}
\newtheorem{assumption}[assume]{Assumption}

\theoremstyle{definition}

\theoremstyle{remark}
\newtheorem{remark}[theorem]{Remark}

\newcommand{\with}{\, ; \,}
\newcommand{\given}{\, | \,}

\newcommand{\rbb}{\mathbb{R}}

\newcommand{\ep}{\epsilon}
\renewcommand{\d}{\mathrm{d}}
\newcommand{\ddt}{\frac{\d}{\d t}}

\newcommand{\vf}{\varphi}
\newcommand{\f}{f}
\newcommand{\F}{F}
\newcommand{\Var}{\mathrm{Var}}
\newcommand{\Med}{\mathrm{Med}}
\newcommand{\gini}{\mathrm{Gini}}

\newcommand{\income}{\mathcal{I}}

\newcommand{\mob}{\mathcal{M}}

\newcommand{\lognorm}[2]{\mathrm{LogNorm}\!\left(#1,#2\right)}
\newcommand{\qrs}{\mathrm{QRS}}
\newcommand{\gic}{\mathrm{GIC}}
\newcommand{\gicrs}{\mathrm{GICRS}}
\newcommand{\cv}{V}

\title{Inequality and Mobility in a Minimal Model for Evolving Income Distributions}
\author{Scott A.~McKinley\thanks{Department of Mathematics, Tulane University, New Orleans, LA.} and Gary A.~Hoover\thanks{Murphy Institute for Political Economy, Tulane University, New Orleans, LA.}\phantom{..}\thanks{Department of Economics, Tulane University, New Orleans, LA.}}

\begin{document}

\maketitle

\begin{abstract}
In this paper we explore the dynamic relationship between income inequality and economic mobility through a pairing of a population-scale partial differential equation (PDE) model and an associated individual-based stochastic differential equation (SDE) model. We focus on two fundamental mechanisms of income growth: (1) that annual growth is percentile-dependent, and (2) that there is intrinsic variability from one individual to the next. Under these two assumptions, we show  that increased economic mobility does not necessarily imply decreased income inequality. In fact, we show that the mechanism that directly enhances mobility, intrinsic variability, simultaneously increases inequality. 

Using Growth Incidence Curves, and other summary statistics like mean income and the Gini coefficient, we calibrate our model to US Census data (1968-2021) and show that there are multiple parameter settings that produce the same growth in inequality over time. Strikingly, these parameter settings produce dramatically different mobility outcomes. Naturally, the greater disparity there is between annual percentage growth in the upper and lower income levels, the less ability there is for individuals to climb the percentile ranks over a fixed period of time. However, more than this, the model shows that whatever mobility does exist, it decreases substantially over time. In other words, while it may remain true that opportunity to reach the upper ranks mathematically persists in the long run, that long run gets longer and longer every year. 
\end{abstract}

\section{Introduction}

In the thorough analysis done by \citet{piketty2003income,piketty2014inequality}, the authors showed that income inequality in the United States has been steadily over the last fifty years and also grew in the pre-war era dating back to at least 1913. While it was a tremendous scholarly achievement to quantify these facts, the observation itself should hardly be considered novel. The effects of inequality are pervasive and continue to worsen. Take, for example, the data displayed in \Cref{fig:quantile-limits}, which shows that since 1968, individuals in the top 5\% of the income distribution have seen their incomes grow at a rate that dwarfs those of the rest of the population. Relative to inflation, the gains at the top stand in stark contrast to the little to no gain experienced below. In fact, through their novel use of IRS data, Picketty and Saez showed that during the economic expansion occurring between 2002 and 2007, that two-thirds of the nation's total income gains went to the top 1\% of income earners. 

\begin{figure}[ht]
\includegraphics[width = \textwidth]{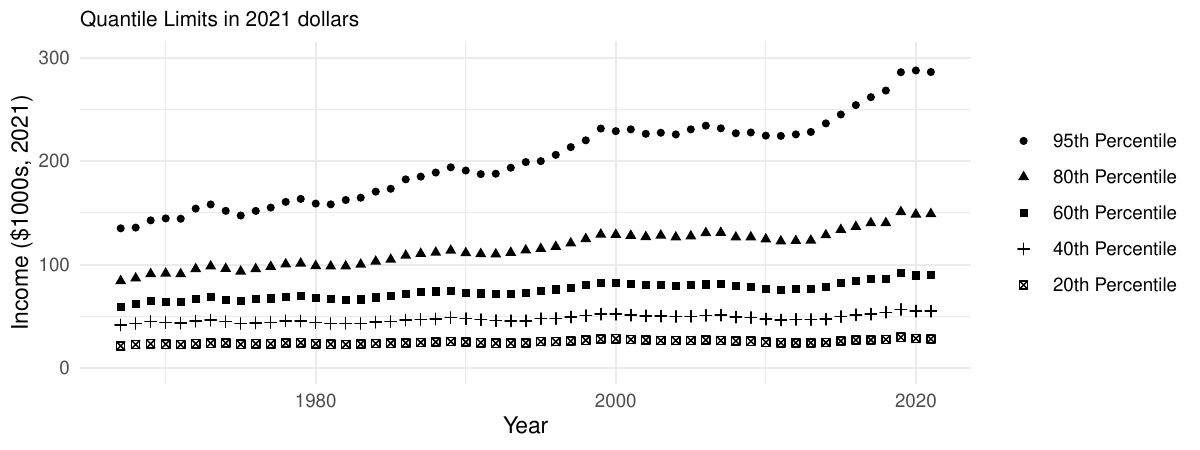}
\caption{\textsc{Quantile upper limits in 2021 dollars.} After adjusting for inflation we see how rapidly the incomes of the high-earners are growing compared to other groups.
\label{fig:quantile-limits}}
\end{figure}

While the fact of inequality is generally recognized and accepted, the mechanisms of inequality and causes for underlying tolerance of inequality remain active subjects of inquiry. One major theme, through theoretical, empirical, and psychological studies, is the complementary role of \emph{economic mobility} -- both in terms of individuals' actual and perceived abilities to traverse the income ranks. One powerful framework for considering tolerance for inequality is Rawls's Veil of Ignorance \citep{rawls1971justice}. To what extent would an individual in society be comfortable with high levels of inequality if their placement in the ranks was unknown? The question can be considered in terms of risk assessment and utility \citep{liang2017optimal}, but this may not fully explain acceptance of income distributions with high degrees of inequality. How might the calculus change if individuals believe that there is a credible chance of climbing the income ladder, whether by luck or by skill \citep{cohen2015immobility}? In a mobile economy, the higher reaches of the income distribution might feel accessible, meaning that otherwise intolerable distributions might seem more acceptable.

Beginning with the seminal study by \citet{benabou2001social}, the prospect for upward mobility (POUM) has been shown to be a rational mollifier of attitudes toward inequality and the need for redistribution. But psychological effects, like an individual's inequality climate, can also affect perceptions of fairness in redistribution policies \citep{shariff2016income,bernasconi2023income}, and the expectation (and reality) of future income can be affected by intergenerational effects \citep{toney2021intergenerational}. Meanwhile, the perception of the need for redistribution can be undermined by heterogeneous mobility within subgroups \citep{akee2019race} and profound underestimation of racial disparities \citep{akee2019race,davidai2022americans}.

Notably, the term ``mobility'' has taken on a variety of meanings over the years (see, for example,  \citep{bartholomew1967stochastic,conlisk1990monotone,fields1999measurement,kopczuk2010earnings}). Naturally, definitions are often tailored to the nature of the available data or mathematical model being considered. But these various definitions can then contribute to different interpretations of the state of play. As observed by \citet{jenkins2006trends}, empirical studies conducted under different notions of mobility can lead to opposing conclusions, therefore leading to differing policy prescriptions. For example, in their study of earnings inequality and mobility in the United States, \citet{kopczuk2010earnings} introduced the following ``inequality-only'' definition of mobility:
\begin{equation}
\begin{aligned}
    \text{\citep{kopczuk2010earnings}:} \quad& \text{Long-term earnings inequality} \\
    & \qquad = \text{Short-term earnings inequality} \times (1 - \text{Mobility}).    
\end{aligned}
\end{equation}
When stated in this way, mobility is seen to be purely a function of change in inequality. \emph{However, a central result of the work we present here is that -- even in the presence of minimal assumptions -- different model inputs can lead to the same inequality dynamics while having dramatically different re-ranking of individuals within the population.}

One reason \citet{kopczuk2010earnings} can conflate change in inequality with economic mobility is that there is a clear empirical inverse relationship between them, demonstrated most clearly for Western economies by the Great Gatsby Curve introduced by \citet{corak2013income} and named by \cite{krueger2012Gatsby}. But the introduction of mathematical models allows for disambiguation of cause and effect, and from this perspective, the relationship between inequality and mobility is emergent rather than prescribed. 

One of the original decompositions of mobility was posited by \citet{creedy2002income}. In that work, the authors recognized that some ability for individuals to move across percentiles can come from a simple compression of the income distribution. (Lower inequality is then associated with higher mobility.) Alternatively, mobility can also result from simple randomness in year-to-year income outcomes. This has also been called earnings variability, e.g.~\citet{kopczuk2010earnings}. The relationship between inequality and mobility due to randomness is not immediately clear, and part of our goal in this work is to establish a model framework in which evaluations of this kind are possible.

The existence of simply articulated models invites questions at the very heart of policy perspectives: do we work to enhance mobility in order to take on inequality \citep{carroll2016mobility}; or do we tackle inequality to enhance mobility \citep{mishel2016ford}? The view we present here is to disentangle the ambiguities of the term ``mobility'' from our ability to speak plainly in terms of input features for a mathematical model. Namely, we consider (1) average annual percentage growth rates (expressed as a function of percentile rank), and (2) log-normal randomness in year-to-year income outcomes. Our fundamental finding is that increasing mobility through mechanism (2) leads to greater inequality over time. On the other hand, manipulating mechanism (1) leads to an inverse relationship between inequality and mobility. That is to say, flattening the percentile-dependent growth curve while holding variability constant reduces inequality while increasing mobility (as expected by \citet{creedy2002income}). In addition, we refine, in a dynamic sense, the usefulness of Growth Incidence Curves, which allow us to pinpoint where changes in the income distribution are occurring, which in turn identifies the causes of overall growth in inequality. Our work can be seen as complementary to the investigations by \citet{jenkins2006trends} or \citet{bergstrom2022role} -- studies in which the authors looked to decompose observed changes in the income distribution into its direct effects on the movements of individuals within the larger population. 

In the remainder of this section, we introduce the assumptions of the mathematical model considered in this work; we express the underlying assumptions in terms of mathematically articulated model inputs; and then we describe the metrics we use to first calibrate the model to US data. Finally, we will describe the metrics we use for inequality and mobility, from which we draw our main conclusions.

\subsection{A model featuring inequality growth and economic mobility}

Consistent with the theme emphasized by \citet{carroll2016income}, we propose and analyze a pair of mathematical models which permit simultaneous analysis of dynamic income inequality and economic mobility. At the full population scale, we introduce a  nonlinear transport partial differential equation (PDE) that models a full income distribution evolving over time. To understand how individuals move within the population, we study a companion stochastic differential equation (SDE) driven by the same model components. These models echo the use of Geometric Brownian motion as a baseline individual dynamic \citep{reed2003pareto}, but with two essential mechanisms of income growth that lead to changes in both overall inequality and the (re-)ordering of individuals within populations. Namely, we include the following model inputs:
\begin{itemize}
\item[-] $R(p)$, a percentile-dependent growth function $(p \in [0,1])$; and 
\item[-] $\cv$, a measure for variability in annual growth among members of the same income cohort.
\end{itemize}
We use $R(p)$ to encode the assertion that the average percentage growth in income is larger for higher-percentile earners than for those in lower percentiles. When $R(p)$ is flat (and we write $R(p) = R_*$), we say that our model is \emph{egalitarian among percentile cohorts}. This is distinguished from a different notion of egalitarianism, which concerns the distribution of income change within a percentile cohort. If $\cv = 0$, all individuals within the same income cohort at time $t = 0$ continue to have identical incomes throughout the simulation. In this case, we say that the model is \emph{egalitarian within cohorts.} When $\cv > 0$, we will say there is \emph{mobility} in the model, in the sense that the randomness of individual outcomes allow for re-ranking among the percentile cohorts. The relationship between parameter regimes and these notions of egalitarianism are summarized in a sketch of our model dynamics contained in \Cref{fig:egalitarian}, and an overview of terminology contained in \Cref{tab:egalitarian}.

\begin{figure}[t!]
    \centering
    \includegraphics[trim={0 0.65cm 0 0},clip, width=0.99\linewidth]{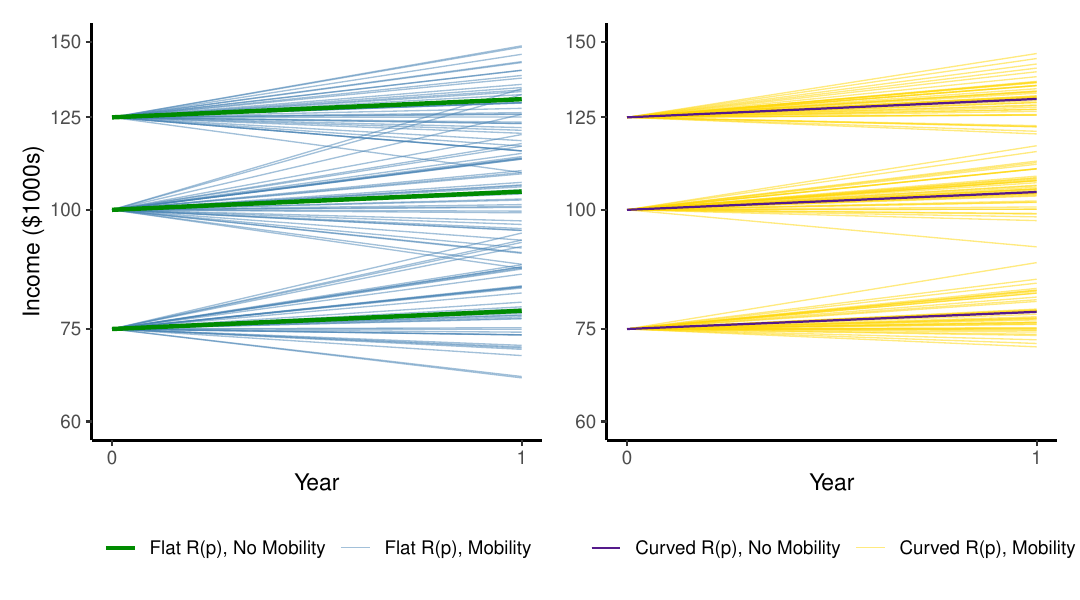}\\
    \caption{A sketch of the one-year income change behaviors that result from our four primary parameter regimes. At left, annual percentage growth has the same mean and variance across all percentiles. At right, the upper percentiles have a slight advantage. In order to calibrate the models so that there is similar growth in inequality, the curved $R(p)$ model requires a smaller degree of year-to-year randomness, leading to lower economic mobility. }
    \label{fig:egalitarian}
\end{figure}
\begin{table}[h] 
    \centering 
    \begin{tabular}{c||c|c}
       & Flat $R(p)$ & Curved $R(p)$\\
    \hline \hline 
    \multirow{3}{*}{$\cv = 0$} & 
    Egalitarian within cohorts & Egalitarian within cohorts \\
    & \& & \& \\
    &  Egalitarian among cohorts & Not egalitarian among cohorts \\
    \hline \multirow{3}{*}{$\cv > 0$} & Not egalitarian within cohorts & Not egalitarian within cohorts \\
    & \& & \& \\
    & Egalitarian among cohorts & Not egalitarian among cohorts \\ 
    \end{tabular}
    \caption{Language used for the various model regimes. ``Cohorts'' refers to percentile groups at the initial time. \label{tab:egalitarian}}
\end{table}

In the flat-$R(p)$ case, we will show that the model inputs $R(p) = R_*$ and $\cv$ have traditional interpretations. In this case, $R_*$ is the average annual percentage growth across the population, and $\cv$/100 is the coefficient of variation (standard deviation divided by the mean) for an individual's annual income growth.  When $R(p)$ is not flat, these descriptions are no longer precisely true, but they remain good intuitive guides for the role these quantities play.

\subsection{Measuring inequality and mobility}

In addition to standard income distribution assessments (mean, median, Gini coefficient), our analysis uses two main tools for assessing model output. The first is an established measure of percentile-dependent growth known as the Growth Incidence Curve (GIC), introduced by \citet{ravallion2003measuring}. The second is a novel measurement of individual mobility that we label \emph{$(P_1 \to P_2)$-mobility}, which assesses the probability that an individual can start at the $P_1$-percentile and be at or above the $P_2$-percentile after a specified number of years. With these tools we will be able to examine more than just the overall year-over-year changes in inequality in a distribution, but also assess the extent to which those changes coincide with unequal growth among cohorts and reordering of individuals within a population over time. 

\subsubsection{Growth Incidence Curves} 
A GIC can be thought of as a dynamic version of a Lorenz curve, which aids in measuring inequality across all percentiles at a given snapshot in time. By assessing relative change in income at every percentile-level simultaneously, GICs can show whether inequality generation is happening in specific tiers of the income distribution, or across the board. 

Let $f(x,t)$ be a probability density function (pdf) summarizing an income distribution (over values $x$) at a time $t$. Let $F(x,t) = \int_0^x f(\xi,t) \d \xi$ be the associated cumulative distribution function (cdf). We define the time-dependent \emph{quantile function} $Q_p(t)$ to be the inverse of the cdf in $x$: $Q_p(t) := F^{-1}(p,t)$. It can be more conventionally written in terms of either the pdf or the cdf as the value satisfying the equation
\begin{equation}
    \Big( p = \int_0^{Q_p(t)} \!\! f(x,t) \d x \Big) \text{ or } \Big( p = F(Q_p(t),t) \Big).
\end{equation}
The function $Q_p(t)$ can be read as the \emph{p}-quantile or ``the income-level of the $100p$-th percentile'' at time $t$. With this function in hand we can define the GIC.

\begin{definition}[Growth Incidence Curve (GIC)]
For a given time $t$ and $p \in (0,1)$, the function $\gic(p,t)$ is the annual percentage growth rate of the $100p$th percentile over the following year:
\begin{equation}
    \gic(p,t) = 100 \left(\frac{Q_p(t + 1)}{Q_p(t)} - 1\right).
\end{equation}
\end{definition}
\begin{figure}[h!]
    \centering
    \includegraphics{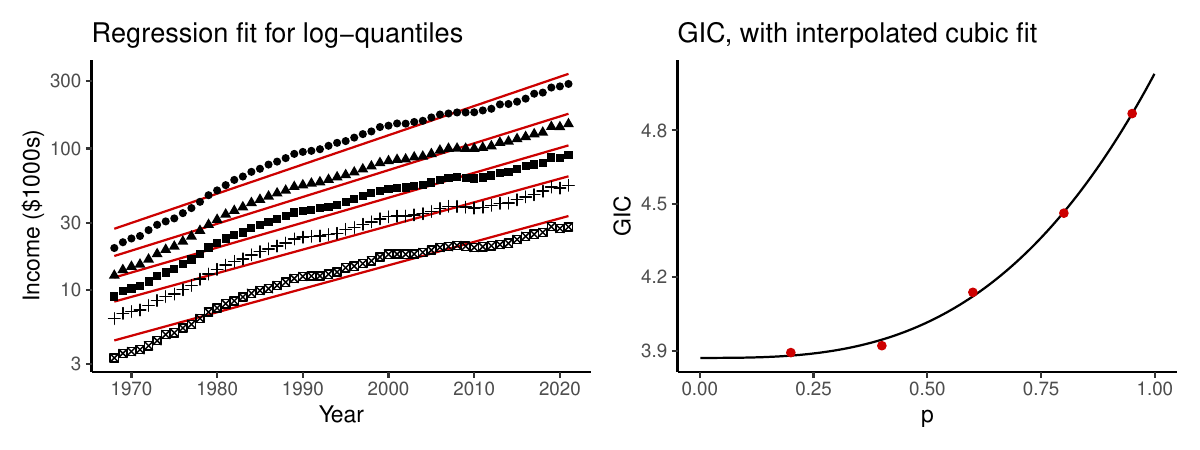}
    \caption{US Census data, 1968-2021: Log-quantiles and Growth Incidence Curve. \textbf{Left.} Using regression, we fit the average annual percentage growth rates for incomes at the 20th, 40th, 60th, 80th, and 95th percentiles. \textbf{Right.} The estimated growth rates are displayed (red circles) and a cubic fit of the form $\gic(p) = \hat{\alpha}_0 + \hat{\delta} p^3$ was computed to estimate interpolated GIC values (black curve, $\hat{\alpha}_0 = 3.94, \hat{\delta} = 1.16$).
    \label{fig:gic-real}}
\end{figure}

Often one wants to assess the average GIC over many years, and multiple methods have been proposed to do this \citep{ravallion2003measuring,bourguignon2011non,fourrier2021bayesian}. For an example of a GIC estimated over a period of time, consider \Cref{fig:gic-real}, in which we present US Census data from 1968-2021. In the left-hand panel, we display time series of (non-inflation adjusted) upper quantile values. Overlaid on these data are least-squares regression lines. The slopes of these lines are good estimates for the GIC evaluated at those percentile levels. In the right-hand panel of \Cref{fig:gic-real}, we show a cubic interpolation of those fit GIC points (red circles) using the formula $\widehat{\gic}(p) = \hat{\alpha}_0 + \hat{\delta} p^3$. The result of the interpolation is the black curve.

While it is natural that the percentile-dependent growth function $R(p)$ has a direct effect on the GIC curves measured from our PDE model, a major finding in this work is that randomness also contributes to the GIC in unexpected ways. When $V > 0$, the model input for percentile-dependent growth rates, $R(p)$, does not equal the model output for the quantile growth rates, $\gic(p)$. 

\subsubsection{($P_1 \to P_2$)-Mobility}

In addition to having an impact on growth in inequality, randomness gives individuals within cohorts the ability to ``climb the income ladder'' relative to their peers. Associated with our PDE model is a naturally defined stochastic differential equation (SDE) agent-based model. In the context of that model, we can compute the probability that individuals can cross percentiles within a specified amount of time. This gives rise to the following notion of mobility.

\begin{definition}[($P_1 \to P_2$)-mobility]
    \label{defn:p1p2-mobility}
    Let $0 < P_1 < P_2 < 100$ and let $0 \leq T_0 \leq T_1 < T_2$. Then the ($P_1 \to P_2$)-mobility of an income distribution is the probability that an individual starting at the $P_1$-percentile at time $T_1$ will have an income at or above the $P_2$-percentile level at time $T_2$. 

    Written in terms of an agent-based model, let $X(t)$ be the income of a random agent within in a population whose quantiles are given by the functions $\{Q_p(t)\}_{t \geq T_0}$, $p \in (0,1)$. Then the $(P_1 \to P_2)$-Mobility is the conditional probability
    \begin{equation} \label{eq:p1p2-mobility}
        \mob(T_0, T_1, T_2 \with P_1, P_2) := P\Big(X(T_2) \geq Q_{p_2}(T_2) \given X(T_1) = Q_{p_1}(T_1)\Big). 
    \end{equation}
\end{definition}

This agent-based perspective is a more targeted evaluation of mobility than some popular whole-population measures, like the Shorrocks and Bartholomew indices \citep{shorrocks1978measurement,bartholomew1996statistical}. Each of these focuses on the probability that individuals in one income block will change to a different block in a given period of time. These aggregate measures do not track whether the movement was upward or downward, and in the case of the Shorrocks index, it does not incorporate the magnitude of the change. The alternative framing of mobility we embrace uses the stochastic process theory for first passage times to define mobility in terms of how long it is expected for an individual to pass from one quintile to another \citep{conlisk1990monotone}. However, for continuum-valued models like the one we present here, computing mean first-passage times can be difficult. Even in the simplest setting, it is necessary to either solve a PDE or to take a numerical average over many individual-income trajectories.  By contrast, the PDE model we propose intrinsically solves the probability distribution of a population given an initial condition. So the $(P_1 \to P_2)$-mobility metric is a natural output of the model. Even better, when average annual growth is constant across all percentiles ($R(p) = R_*$), $(P_1 \to P_2)$-mobility can be solved explicitly.

\subsection{Overview of this work}

In \Cref{sec:model}, we define our PDE and SDE models for the time evolution of income distributions subject to the two dynamic mechanisms of mean percentile-dependent growth $R(p)$ and income-rank mixing due to randomness $\cv$. In \Cref{sec:results:analytical}, we conduct an analysis of the model under two circumstances: first, when $R(p)$ is constant, or \emph{flat} across all percentiles, and second, when $R(p)$ is \emph{curved}. 

\begin{table}[h!]
    \centering
    \begin{tabular}{c|c|c|c}
        Model name & $R_*$ & $\delta$ & $V$   \\
        \hline \hline
        Flat $R(p)$, No Mobility & 4.74 & 0 & 0 \\
        Flat $R(p)$, With Mobility & 4.74 & 0 & 8.7 \\
        Curved $R(p)$, No Mobility & 4.03 & 1.19 & 0  \\
        Curved $R(p)$, With Mobility & 4.29 & 1.15 & 5.0 \\
    \end{tabular}
    \caption{\textsc{Dynamic model parameters}. The form for the annual percentage growth function is assumed to be $R(p) = R_* + \delta p^3$, where $p \in (0,1)$. In all models we used a log-normal distribution with Mean/Gini fit for the 1968 US Census data as the initial condition. Motivation for these choices are derived in \Cref{sec:results:analytical} and the selections are further explained in \Cref{sec:numerical}.
    \label{tab:parameters-dynamic}}
\end{table}

While the flat-$R(p)$ model permits explicit formulas for various income distribution properties, in the curved-$R(p)$ case we are more restricted. Nevertheless we can learn enough to support a targeted numerical investigation of the curved $R(p)$ model with mobility. We discuss the necessary steps for parameterizing multiple models in Sections \ref{sec:numerical:init} and \ref{sec:numerical:dynamic} and show that multiple model settings can produce essentially equivalent growth in inequality. However, in Sections \ref{sec:analysis:flat-rp:mobility} and \Cref{sec:numerical:mobility}, we show that the interplay of mobility and inequality is nuanced even in this minimalist modeling framework.

\section{A PDE model for evolving income distributions}
\label{sec:model}

The two fundamental quantities that define our model are percentile-dependent growth and year-to-year randomness, which is assumed to scale with the annual percentage growth. The model inputs $R(p)$ and $\cv$ are expressed in terms of annual percentage growth properties. Because our model is continuous in time, we must define \emph{instantaneous} versions of both that yield (approximately) the correct one-year properties. 

\begin{definition}[Growth and Variation]
\label{defn:growth-mixing}
    Assume that $R: [0,1] \to \rbb_+$ is a non-decreasing and continuously differentiable function with $R(0) > -100$. Assume further that the constant $\cv$ is non-negative. Given $R(p)$ and $\cv$, we define the infinitesimal mean function $\alpha(p)$ and infinitesimal variance parameter $\beta$ by
    \begin{equation} \label{eq:defn-alpha-beta}
    \begin{aligned}
        \alpha(p) &:= \ln\!\Big(1 + \frac{R(p)}{100}\Big), \text{ and } \\
        \beta &:= \frac{1}{2} \ln\!\Big(1 + \left(\frac{\cv}{100}\right)^{\!2}\Big).
    \end{aligned}
    \end{equation}
\end{definition}
The motivation for this definition will be spelled out in \Cref{sec:flat-rp} (\Cref{cor:one-year-lognormal}) where $R(p)$ is assumed to be a constant $R_*$. 

\begin{definition}[Income distribution: PDE model]
\label{defn:pde-model}
    Let $R(p)$, $\cv$, $\alpha(p)$, and $\beta$ be given as in \Cref{defn:growth-mixing}. Let $\f_0(x)$ be a probability density function (pdf) that is continuously differentiable, supported on the positive real line, satisfying
    \begin{equation} 
    \label{eq:pdf-initial-assumptions}
        \begin{aligned}
            \lim_{x \to 0} x^2 \f_0'(x) = \lim_{x \to \infty} x^2 \f'_0(x) = 0
        \end{aligned}
    \end{equation}

    We define our time-dependent income distribution to be the solution (when it exists) of the PDE
    \begin{equation} \label{eq:defn-pde-pdf}
    \begin{aligned}
        \partial_t \f(x,t) + \partial_x \big(\alpha(\F(x,t)) x \f(x,t)\big) &= \partial_{xx} \big(\beta x^2 \f(x,t) \big), & x > 0, t > t_0; \\
        \f(x,0) &= \f_0(x), & x \geq 0; \\
        \f(0,t) &= 0, & t > t_0.
    \end{aligned}
    \end{equation}
    where $\F(x,t) \coloneqq \int_0^x \f(\xi, t) \d \xi$ is the cumulative distribution function (cdf).
\end{definition}
Any solution of \Cref{eq:defn-pde-pdf} under the given assumptions will conserve mass in the sense that $\int_0^\infty \f(x,t) \d x = 1$ for all $t \geq t_0$. At first glance, the PDE for $\f$ looks like it depends on an auxiliary equation, which would describe the dynamics of the cdf, $\F(x,t)$. One can write \eqref{eq:defn-pde-pdf} exclusively in terms of $\f$ by introducing convolution with a Heaviside function (see \cite{jourdain2013propagation} for example), but in looking at the PDE satisfied by the cdf, we see that $F$ can be studied autonomously: 
\begin{equation} \label{eq:defn-pde-cdf}
\begin{aligned}
    \partial_t \F(x,t) + \big(\alpha(\F(x,t)) - 2 \beta\big) x \, \partial_x \F(x,t) &= \beta x^2 \partial_{xx} \F(x,t), \hspace*{0.2 in} & x > 0, \,t > t_0;\\
    F(x,t_0) &= \int_0^x \f_0(\xi) \d \xi, & x \geq 0;\\
    F(0,t) &= 0, & t > t_0; \\
    \lim_{x \to \infty} F(x,t) &= 1 & t > t_0.
\end{aligned}
\end{equation}
In \Cref{app:numerical} we show the derivation of the PDE for the cdf of the log-scale version of this problem and discuss how it is used as the basis for our numerical simulations.

Partial differential equations of the form \eqref{eq:defn-pde-pdf} arise naturally from systems of interacting agents (or particles) in which the dynamics of individuals are affected by the relative distribution of other individuals in the system. Though originally developed with physics applications in mind, these interacting particle systems are widely used to model capital markets (see \citet{fernholz2002stochastic} for an overview  of stochastic portfolio theory, for example) and in (mean field) game theory dynamics (e.g.~\citet{carmona2015probabilistic} and \citet{lacker2016general}). The existence and uniqueness theory for \Cref{eq:defn-pde-cdf} with a specific type of function $\alpha(p)$ was first demonstrated by \citet{jourdain1997diffusions}. Subsequently, \citet{jourdain2013propagation} extended this theoretical work to certain cases with $p$-dependent $\beta$.

From a modeling point of view, the income dynamics can best be understood from an individual \emph{agent-based} perspective. This connection was explored by both \citet{jourdain1997diffusions} and \citet{jourdain2013propagation}. When the percentile-dependent growth function $R(p)$ is constant, there is an immediate connection between the PDEs \eqref{eq:defn-pde-pdf} and \eqref{eq:defn-pde-cdf} and the Geometric Brownian motion. It is natural to generalize this observation to non-constant $R(p)$ and assert the following agent-based model for incomes within an interacting system.

\begin{definition}[Agent-based model]
\label{defn:agent-based-model}
    Let $R(p)$, $\cv$, $\alpha(p)$, and $\beta$ be given as in \Cref{defn:growth-mixing}. Let $\f_0(x)$, $f(x,t)$ and $F(x,t)$ be given as in \Cref{defn:pde-model}. Let $x_0$ be a value in the support of the distribution $\f_0(x)$.\footnote{A point $x_0$ is in the support of a probability density function $\f(x)$ if, for any $\ep > 0$, $\displaystyle \int_{x_0 - \ep}^{x_0 + \ep} \f(x) \d x > 0$.}

    Then the dynamics $X(t)$ of an individual income that begins at the level $x_0$ at time $t_0$ can be modeled by the SDE 
    \begin{equation} \label{eq:defn-sde-X}
    \begin{aligned}
        \d X(t) &= \alpha\big(F(X(t),t)\big) X(t) \d t + \sqrt{2 \beta} X(t) \d W(t), \quad t \geq t_0; \\
        X(t_0) &= x_0.
    \end{aligned}
    \end{equation}
\end{definition}
Interaction of individuals in a population through their relative ranks introduces a host of technical issues that we do not address here. This is why we state the relationship between the foregoing agent-based model and the associated PDE models as a definition rather than a theorem. For a more thorough discussion of mathematical results we refer the reader to specific rank-based models, like the Atlas model \citep{banner2005atlas, ichiba2011hybrid, ichiba2013strong}, that arise naturally in stochastic portfolio theory \citep{fernholz2009stochastic} or work in a more general direction undertaken by \citet{shkolnikov2012large}, \citet{jourdain2013propagation}, and  \citet{zhang2021topics}. These latter works tend to rely on assumptions about the initial condition \citep{shkolnikov2012large} or the shape of $\alpha$ \citep{jourdain2013propagation}, and so the full mathematical development of systems of this type remains an open project.

\section{Analytical Results}
\label{sec:results:analytical}

Before exploring numerical experiments, we collect a series of analytical results that inform our investigation. Explicit solutions of the PDE models \eqref{eq:defn-pde-pdf} and \eqref{eq:defn-pde-cdf} exist when $R(p)$ is flat and/or when $\cv = 0$. The methods of solution differ, so we separate the analyses into different subsections. When $R(p)$ is not flat (or \emph{curved} as we will refer to it) and $\cv > 0$, there is no explicit formula for solutions and we must appeal to previous work for the theory of existence and uniqueness \citep{jourdain1997diffusions,jourdain2013propagation}. 

In the flat-$R(p)$ case with log-normal initial data, log-normality is preserved (\Cref{prop:lognormal-main}) and there are explicit formulas for certain time-dependent quantities of interest like the mean, $p$-quantile values, the Gini, and $(P_1 \to P_2)$-mobility. In the curved-$R(p)$ setting, we can express ordinary differential equations (ODEs) that the quantities satisfy. These ODEs yield sufficient insight to help parameterize numerical investigations that follow. 

For quick reference, we remind the reader that a random variable $X \sim \lognorm{\mu}{\sigma^2}$  can be written $e^{\mu + \sigma Z}$ where $Z \sim \text{Norm}(0, 1)$. The log-normal density is given by the function
\begin{equation} \label{eq:lognorm-density}
    f_{LN}(x \with \mu, \sigma) = \frac{1}{x \sqrt{2 \pi \sigma^2}} \,\, e^{\displaystyle -\frac{(\ln(x) - \mu)^2}{2 \sigma^2}}.
\end{equation}
The mean, median, variance, and Gini coefficient of $X$ are given by the formulas provided in \Cref{table:lognorm}.
\begin{table}[h!]
    
    \centering
    \begin{tabular}{cccc}
        \uline{Mean} & \uline{Median} & \uline{Variance} & \uline{Gini}       \\
        $\displaystyle e^{\,\mu + \frac{\sigma^2}{2}}$ & 
        $\displaystyle e^{\,\mu}$ 
        & $\displaystyle \Big(e^{\sigma^2} - 1\Big) e^{2 \mu + \sigma^2}$
        & $\displaystyle 2 \Phi\Big(\frac{\sigma}{\sqrt{2}}\Big)-1$
    \end{tabular}
    \caption{Basic properties of the $\lognorm{\mu}{\sigma^2}$ distribution. Here, $\Phi(x)$ is the cumulative distribution function of a standard normal random variable
    \label{table:lognorm}}
\end{table}

Note that because of the way the log-normal distribution is skewed, the mean is always at least as large as the median. Moreover, the variance of a log-normal random variable can be written in terms of its mean ($E(X)$) and median ($\Med(X)$):
\begin{equation}
\begin{aligned}
    \Var(X) &= \left(\Big(\frac{E(X)}{\Med(X)}\Big)^2 - 1\right).  
\end{aligned}
\end{equation}
From this relationship we can see that if the median of a log-normal random variable equals its mean, then it has variance zero, i.e. it is a deterministic constant.

\begin{assumption}[Log-Normal Initial Distribution] \label{a:lognormal-init}
Under this assumption, we take the initial pdf $f_0(\cdot, t_0) \sim \lognorm{\mu_0}{\sigma_0^2}$ for some $\mu_0 \in \rbb$ and $\sigma_0 \geq 0$.  
\end{assumption}

\subsection{Flat $R(p)$.}
\label{sec:flat-rp}

When $R(p)$ is constant for all $p \in (0,1)$, or in other words, egalitarian among income cohorts, our model satisfies Gibrat's Law in that the mean percentage growth and variance are the same across all income levels. 

\begin{assumption}[Flat-$R(p)$] \label{a:flat-rp}
The dynamics of the income distribution satisfy \Cref{eq:defn-pde-pdf} with constant $R(p) = R_*$. The corresponding infinitesimal mean as $\alpha = \ln(1 + R_*/100)$. 
\end{assumption}

Our first result follows from viewing the evolving income distribution as the law of a geometric Brownian motion, which is the agent-based model associated with adopting Assumptions \ref{a:lognormal-init} and \ref{a:flat-rp}. See \citet{reed2003pareto} for another development of geometric Brownian motion as the base for an income inequality model.

\begin{proposition}[Flat $R(p)$ preserves log-normality] \label{prop:lognormal-main}
    Under Assumptions \ref{a:lognormal-init} and \ref{a:flat-rp}, the agent-based model $X(t)$ given in \Cref{defn:agent-based-model} has the exact solution
    \begin{equation} \label{eq:X-gbm}
        X(t) = X(t_0) e^{(\alpha - \beta) (t-t_0) + \sqrt{2 \beta} W(t - t_0)}
    \end{equation}
    where $W$ is a standard Brownian motion. Therefore 
    \begin{equation} \label{eq:X-lognormal}
        X(t) \sim \mathrm{LogNorm}\big(\mu(t - t_0), \sigma^2(t - t_0)\big),
    \end{equation}
    where
    \begin{equation} \label{eq:X-lognormal-param}
        \mu(t) = \mu_0 + (\alpha - \beta) t \quad \textrm{and} \quad 
        \sigma^2(t) = \sigma^2_0 + 2 \beta t.
    \end{equation}
    Moreover, the functions $f(x,t)$ and $F(x,t)$ given by \Cref{defn:pde-model} are the pdf and cdf of the $\mathrm{LogNorm}\big(\mu(t), \sigma^2(t)\big)$ distribution.
\end{proposition}
\begin{proof}
    The exact solution for the geometric Brownian motion \eqref{eq:X-gbm} is a standard result \citep{reed2003pareto,oksendal2013stochastic}. The log-normality of $X(t)$ under a log-normal initial condition follows from observing three equivalences in distribution (denoted by the symbol `$\approx$'). Namely, $X_0 \approx \exp(\mu_0 + \sigma_0 Z_0)$ and $W(t - t_0) \approx \sqrt{t-t_0} Z_1$ where $Z_0, Z_1 \sim \mathrm{Norm}(0,1)$. Moreover, for any two constants $a$ and $b$, $a Z_0 + b Z_1 \approx \sqrt{a^2 + b^2} Z$ where $Z \sim \mathrm{Norm}(0,1)$. We have \begin{displaymath} \begin{aligned}
            X(t) &\approx e^{\mu_0 + (\alpha - \beta) (t-t_0) + \sigma_0 Z_0 + \sqrt{2 \beta (t-t_0)} Z_1} \\
            &\approx e^{\mu_0 + (\alpha - \beta) (t-t_0) + \sqrt{\sigma_0^2 + 2 \beta (t-t_0)} Z}.
        \end{aligned}
    \end{displaymath}
    By definition, this means $X(t)$ is log-normal for all $t \geq t_0$, with parameters given in \eqref{eq:X-lognormal} and \eqref{eq:X-lognormal-param}.
\end{proof}

\begin{corollary}
\label{cor:lognormal-mmg}
    Under Assumptions \ref{a:lognormal-init} and \ref{a:flat-rp}, the mean, median, and Gini are
    \begin{equation} \label{eq:X-lognormal-mmg}
    \begin{aligned}
        E(X(t)) &= \exp\Big(\mu_0 + \frac{\sigma^2_0}{2} + \alpha (t-t_0) \Big); \\
        \Med(X(t)) &= \exp\Big(\mu_0 + (\alpha - \beta) (t-t_0) \Big); \textrm{ and} \\
        \gini(X(t)) &= 2\Phi\left(\sqrt{\frac{1}{2}(\sigma^2_0 + 2 \beta (t-t_0))}\right) - 1.
    \end{aligned}    
    \end{equation}
    The median formula is a special case of the $100 p$-percentile result:
    \begin{equation} \label{eq:lognormal-percentile}
        Q_p(t) = \exp\Big(\mu(t - t_0) + \sigma(t - t_0) \zeta_p  \Big),
    \end{equation}
    where $\mu$ and $\sigma$ are given by \eqref{eq:X-lognormal-param} and $\zeta_p \coloneqq \Phi^{-1}(p)$. That is to say, $\zeta_p$ is the value such that a standard normal random variable $Z$ satisfies
    \begin{equation} \label{eq:Z-percentile}
        P(Z \leq \zeta_p) = p.
    \end{equation}
\end{corollary}
\begin{proof}
    The mean, median, and Gini formulas follow directly from the standard log-normal properties recorded in \Cref{table:lognorm}. To derive the quantile property, note that 
    \begin{displaymath}
        \begin{aligned}
            P(X(t) \leq k) &= P\big(\exp(\mu(t-t_0) + \sigma(t-t_0) Z) \leq k\big)\\
            &= P\left(Z < \frac{\ln(k) - \mu(t-t_0)}{\sigma(t-t_0)}\right).
        \end{aligned}
    \end{displaymath}
    It follows that $P(X(t) < k) = p$ if and only if $(\ln(k) - \mu(t))/\sigma(t) = \zeta_p$. Solving for $k$ yields \eqref{eq:lognormal-percentile}. For the median formula, note that $\zeta_{0.5} = 0$.
\end{proof}

From the previous equations, we can justify the relationship between the model inputs $R(p)$ and $\cv$ and the infinitesimal mean and variance $\alpha(p)$ and $\beta$ given in \Cref{defn:growth-mixing}. We show that the assumption implies that both the mean and the standard deviation of annual growth are independent of the income percentile, and $\cv/100$ is a constant of proportionality between the mean and standard deviation of annual growth.

\begin{corollary} \label{cor:one-year-lognormal}
    Let $R(p) = R_*$ be a constant, and let $\cv > 0$ be given. Define $\alpha$ and $\beta$ as in \Cref{defn:growth-mixing}. Suppose further that $X(t)$ satisfies the agent-based model with initial condition set to the $100p$-percentile at time $t_0 = 0$, i.~e.~$X(0) = e^{\mu_0 + \sigma_0 \zeta_p}$. Then
    \begin{equation} \label{eq:one-year-lognormal}
    \begin{aligned}
        E\Big(\frac{X(1)}{X(0)}\Big) &= 1 + \frac{R_*}{100}; \\
        SD\Big(\frac{X(1)}{X(0)}\Big) &= \frac{\cv}{100} \Big(1 + \frac{R_*}{100}\Big).
    \end{aligned}
    \end{equation}
\end{corollary}
\begin{proof}
    The exact solution for $X(t)$ is
    \begin{displaymath}
        X(1) = X(0) e^{\alpha - \beta + \sqrt{2 \beta} W(1)}.
    \end{displaymath}
    By definition of $\alpha$, \Cref{eq:defn-alpha-beta},
    \begin{displaymath}
        E(X(t)) = X(0) e^{\alpha} = X(0) (1 + R(p)/100),
    \end{displaymath}
    which proves the first part of \eqref{eq:one-year-lognormal}. Meanwhile, 
    \begin{displaymath}
    \begin{aligned}
        \Var(X(t)) &= X(0)^2 \big(e^{2 \beta} - 1\big) e^{2 \alpha} \\
        &= X(0)^2 \Big(e^{\ln\big(1 + (\cv/100)^2\big)} - 1\big) \Big(1 + \frac{R_*}{100}\Big)^2,
    \end{aligned}
    \end{displaymath}
    which confirms the second part of \eqref{eq:one-year-lognormal}. 
    \end{proof}

\subsubsection{Growth Incidence Curves when log-normality is preserved}
\label{sec:analysis:flat-rp:gic}

At first glance, one might expect the GIC to simply match $R(p)$, since both are measures of the annual percentage growth. However, $R(p)$ is a property of individuals in the population, whereas $\gic(p)$ is a property of the quantile values. If we spoke of 100 individuals, then $R(0.2)$ addresses the future income of the person sitting in the 20th position at the initial time.  By contrast $\gic(0.2)$ examines how the income associated with the person in the 20th position at time zero compares to the person in the 20th position at time one.

\begin{proposition}
\label{prop:lognormal-gic}
    Under Assumptions \ref{a:lognormal-init} and \ref{a:flat-rp}, 
    \begin{equation} \label{eq:gic-lognormal}
        \gic(p) = R_* + 100 \left(1 + \frac{R_*}{100}\right) \left(\frac{\exp\Big(\Big(\sqrt{1 + \frac{1}{\sigma_0^2}\ln(1 + (\cv/100)^2)} - 1\Big) \sigma_0 \zeta_p\Big)}{\sqrt{1 + (\cv/100)^2}} - 1\right) 
    \end{equation}
\end{proposition}
\begin{proof}
    From \Cref{cor:one-year-lognormal}, recall the relations
    \begin{displaymath}
        e^{\alpha} = \Big(1 + \frac{R_*}{100}\Big) \text{ and } e^{2 \beta} = \Big(1 + \Big(\frac{\cv}{100}\Big)^2\Big).
    \end{displaymath}
    Meanwhile, consider the stochastic process $\widetilde{X}(t)$ that satisfies the SDE \eqref{eq:defn-sde-X} with constant coefficients and initial condition $\widetilde{X}(0) = Q_p(0)$. Then
    \begin{displaymath}
        E\big(\widetilde{X}(1)\big) = Q_p(0) E\big(\exp(\alpha - \beta + \sqrt{2 \beta} Z)\big) = Q_p(0) e^{\alpha}. 
    \end{displaymath}
    Together, the two preceding equations imply that
    \begin{displaymath}
        R_* = 100 \left(\frac{E\big(\widetilde{X}(1)\big)}{Q_p(0)} - 1\right).
    \end{displaymath}
    By definition of the $\gic$ and \Cref{eq:lognormal-percentile}, we have
    \begin{displaymath}
    \begin{aligned}
        \gic(p) &= 100 \left(\frac{Q_p(1)}{Q_p(0)} - 1\right) \\ 
        &= R_* + 100 \left(\frac{Q_p(1)}{Q_p(0)} - \frac{E\big(\widetilde{X}(1)\big)}{Q_p(0)}\right) \\
        &= R_* + 100 \, e^\alpha \Big(\exp\Big( \Big(\sqrt{1 + 2 \beta / \sigma_0^2} - 1\Big) \sigma_0 \zeta_p - \beta \Big) - 1\Big).
    \end{aligned}
    \end{displaymath}
    \Cref{eq:gic-lognormal} follows after substituting the formulas for $e^\alpha$ and $e^{2 \beta}$.
\end{proof}

The key observations from this formula are that (a) despite $R(p)$ being a constant value $R_*$, the GIC is a function of $p$ that may be larger or smaller than $R_*$; and (b) the form of the departure from $R_*$ is basically $C_1 \big(e^{C_2 \zeta_p} - 1\big)$, where $C_1$ and $C_2$ are constants that do not depend on $p$. Since $\lim_{p \to 0} \zeta_p = -\infty$ and $\lim_{p \to 1} \zeta_p = \infty$, the second observation implies that the assumption of proportional infinitesimal variance may not be realistic at the extremes $p = 0$ and $p = 1$.

It is important to note again the distinction that $R(p)$ is considered a model input, while $\gic(p)$ is considered a model output. In Section \ref{sec:numerical} we will use GIC results to parameterize the PDE model under various assumptions about $R(p)$.

\subsubsection{The mobility/inequality trade-off when log-normality is preserved}
\label{sec:analysis:flat-rp:mobility}

In each of our models, we seek to quantify the mobility of individuals within a population, subject to constraints set on the full income distribution. In the flat-$R(p)$ model with log-normal initial condition, a change in the randomness parameter $\cv$ does not affect the growth of mean income (recall \Cref{cor:lognormal-mmg}). Therefore, we can compare mobility across a continuum of models that have the same mean dynamics by varying $\cv$ while holding the other model parameters $\mu_0$, $\sigma_0$, and $R_*$ fixed. In the following proposition, we report an exact formula for the probability that an individual at the $P_1$-percentile at the initial time is at or above the $P_2$-percentile, $T$ years later. 

\begin{theorem}[$(P_1 \to P_2)$-Mobility]
\label{prop:lognormal-mobility}
    Let $t_0 \geq 0$ be given, and assume log-normality (Assumptions \ref{a:lognormal-init} and \ref{a:flat-rp} hold.) Let $T_1 \geq t_0$ and for $\tau > 0$ let $T_2 = T_1 + \tau$. Then
    \begin{equation} \label{eq:lognormal-mobility}
        \mob(t_0,T_1,T_2 \with P_1, P_2) = 1 - \Phi\!\left(\frac{1}{\sqrt{\rho(T_1) \tau}} \Big(\sqrt{1 + \rho(T_1) \tau}\, \zeta_{p_2} - \zeta_{p_1} \Big) \right)
    \end{equation}
    where
    \begin{displaymath}
        \rho(\tau) = \frac{2 \beta}{\sigma_0^2 + 2 \beta T_1}.
    \end{displaymath}
\end{theorem}
\begin{proof}
Without loss of generality, take the initial time of the system to be $t_0 = 0$. If the initial distribution is $\lognorm{\mu_0}{\sigma_0^2}$, then the population income distributions at times $T_1$ and $T_1 + \tau$ are given by 
\begin{displaymath}
\begin{aligned}
    f(\cdot,T_1) &\sim \lognorm{\mu_0 + (\alpha - \beta) T_1}{\sqrt{\sigma_0^2 + 2 \beta T_1}}, \\
    f(\cdot,T_1 + \tau) &\sim \lognorm{\mu_0 + (\alpha - \beta) (T_1 + \tau)}{\sqrt{\sigma_0^2 + 2 \beta (T_1 + \tau)}}.
\end{aligned}
\end{displaymath}  
Recall that for a standard normal random variable $Z$, we write $\zeta_{p}$ for the value satisfying $p = P(Z \leq \zeta_{p})$.
We can write the initial $100p_1$-percentile at time $T_1$ and the final $100p_2$-percentile at time $T_1 + \tau$ as
\begin{displaymath}
\begin{aligned}
    q_1 &= e^{\mu_0 + (\alpha - \beta)T_1 + \sqrt{\sigma_0^2 + 2 \beta T_1} \zeta_{p_1}}, \\
    q_2 &= e^{\mu_0 + (\alpha - \beta)(T_1 + \tau) + \sqrt{\sigma_0^2 + 2 \beta (T_1 + \tau)} \zeta_{p_2}}.
\end{aligned}
\end{displaymath}
Now, define $\widetilde{X}(t)$ to be the agent-based model for incomes starting at income level $q_1$ at time $T_1$. That is, $\widetilde{X}(t)$ satisfies \Cref{eq:defn-sde-X} with $\widetilde{X}(T_1) = q_1$. Then, in distribution, we have the following equality
\begin{displaymath}
    \widetilde{X}(T_1 + \tau) \stackrel{\mathcal{D}}{=} q_1 e^{(\alpha - \beta) \tau + \sqrt{2 \beta \tau} Z}.
\end{displaymath}
It follows that
\begin{displaymath}
\begin{aligned}
    P\Big(\ln\big(\widetilde{X}(T_1 + \tau)\big) > \ln(q_2)\Big) 
    &= P\Big(\sqrt{2 \beta \tau} Z > \ln(q_2/q_1) - (\alpha - \beta) \tau \Big) \\
    &= P\Big(\sqrt{2 \beta \tau} Z >  \sqrt{\sigma_0^2 + 2\beta (T_1 + \tau)} \zeta_{p_2} - \sqrt{\sigma_0^2 + 2 \beta T_1} \zeta_{p_1}\Big) \\
    &= P\Big(Z > 
    \sqrt{\frac{\sigma_0^2 + 2\beta T_1}{2 \beta \tau}} \Big(\sqrt{1 + \frac{2 \beta \tau}{\sigma_0^2 + 2 \beta T_1}} \zeta_{p_2} - \zeta_{p_1}\Big)\Big),
\end{aligned}
\end{displaymath}
which simplifies to \Cref{eq:lognormal-mobility}.
\end{proof}

\begin{figure}[t!]
    \centering
    \includegraphics[width=0.9\linewidth]{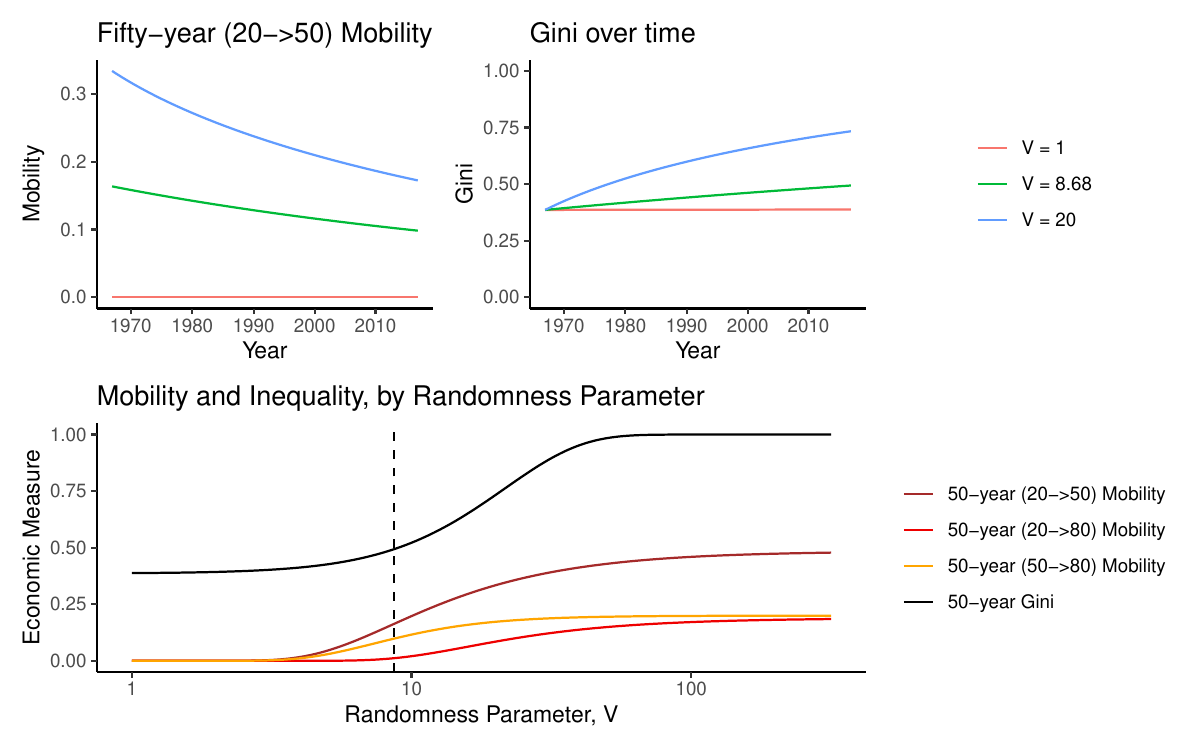}
    \caption{\textsc{Mobility/Inequality when there is egalitarianism among cohorts.} \textbf{(Top)} Fifty-year mobility decreases over time (left) because the Gini coefficient increases over time (right). \textbf{(Bottom)} We display multiple mobility measures for the flat-$R(p)$ model ($R(p) = 4.49$), over a range of values for the randomness parameter $\cv$. Fifty-year mobility increased, but so did the 50th-year Gini coefficient. 
    \label{fig:mobility-inequality}}
\end{figure}

\begin{corollary}
\label{cor:mobility-limits}
    Under the assumptions of \Cref{prop:lognormal-mobility}. We have the following limits:
    \begin{equation}
    \label{eq:mobility-limits}
    \begin{aligned}
        \textrm{for each } T > t_0, \quad \lim_{\tau \to \infty} \mob(t_0, T, T + \tau \with P_1, P_2) &= 1 - p_2; \\
        \textrm{for each } \tau > 0, \quad \lim_{T \to \infty} \mob(t_0, T, T + \tau \with P_1, P_2) &= 0.
    \end{aligned}
    \end{equation}
\end{corollary}

The implications of the corollary are as follows. If it were possible for agents to be infinitely prolonged, then eventually the probability that an agent will be at or above the $P_2$-percentile in income converges to $1-p_2$. In other words, incomes in the log-normal model mix in the long run. However, \emph{the long run takes longer and longer to achieve.} If we think of $\tau$ as a typical agent lifetime, then the probability that an agent is able to traverse the income distribution becomes smaller and smaller over time. This result is visualized in the top left panel of Figure \ref{fig:mobility-inequality}. To generate these mobility curves and the Gini coefficient time-series displayed in the right panel, we started by calculating $\mu_0$ and $\sigma_0$ using a Mean-Gini log-normal fit to US census data for $t_0 = 1968$ (see \cref{tab:lognormal-init} below). For the annual growth parameter, we used the flat-$R(p)$ value $R_* = 4.49$ fit from the mean income growth in the US, 1968-2021. We display the dynamics for multiple choices of the randomness parameter. The specific value $\cv = 8.68$ is the result of a best fit for Gini coefficient growth in the US over the same period. 

In alignment with the $T \to \infty$ limit result from \Cref{cor:mobility-limits}, we see a steady decrease in the 50-year $(20 \to 50)$-mobility under different assumptions for the randomness parameter $\cv$. The reason is that, while a larger randomness parameter leads to greater long-term mobility, the Gini coefficient of the corresponding population is also growing faster. The growth of the Gini coefficient occurs because the Gini coefficient for a log-normal distribution increases with its scale parameter $\sigma$, see \Cref{eq:X-lognormal-mmg}. As observed in \Cref{eq:X-lognormal-param}, the time-dependent scale parameter is an increasing function in $\beta$. So, more short-term mobility interestingly leads to more mobility \emph{and} more inequality.

Considering the dynamic Gini and 50-year mobility together, we see that the increase in inequality outweighs the gains in mobility when agents have a fixed amount of time to make their moves. In the bottom row of \Cref{fig:mobility-inequality}, we display a more general visualization of the effect. The black curve shows how the choice of randomness parameter affects the Gini coefficient 50 years after an initialization based on US 1968 income data. The multi-color curves show the affect of the randomness parameter on 50-year mobility, starting from the same initial data. Values of $V$ that induce mobility, also tend to increase inequality.

\subsection{Curved $R(p)$}
\label{sec:curved-rp}

\begin{proposition}[Curved $R(p)$, no randomness]
\label{prop:curved-rp-nomix}
Let $R(p)$ satisfy the conditions of \Cref{defn:growth-mixing}. Let $\cv = \beta = 0$ and suppose that $\alpha(p)$ be defined as in \Cref{eq:defn-alpha-beta}. Let an initial pdf $f_0(x)$ and cdf $F_0(x)$ satisfy the conditions of \Cref{defn:pde-model}, and suppose that the time dependent pdf $f(x,t)$ and cdf $F(x,t)$ satisfy Eqns.~\ref{eq:defn-pde-pdf} and \ref{eq:defn-pde-cdf}, respectively, with initial time $t_0 = 0$.

Then the quantile process $Q_p(t)$ has the explicit formulation
\begin{equation} \label{eq:curved-rp-nomix-quantiles}
    Q_p(t) = e^{\alpha(p) t} Q_p(0).
\end{equation}
Moreover, for each $t \geq 0$ and $x \geq 0$, let $q(x,t)$ be the solution to the implicit equation
\begin{equation} \label{eq:defn-q}
    x = e^{\alpha(F_0(q)) t} q.
\end{equation}
Then the cdf $F(x,t)$ satisfies
\begin{equation} \label{eq:curved-rp-soln}
    F(x,t) = F(q(x,t),0)
\end{equation}
and (suppressing the dependence of $q$ on $x$ and $t$), the pdf $f(x,t)$ satisfies
\begin{equation}
    f(x,t) = \frac{f_0(q) e^{-\alpha(F_0(q))t}}{1 + \alpha'(F_0(q)) f_0(q) t}.
\end{equation}
\end{proposition}
\begin{proof}
When $\beta = 0$, the PDE for the CDF reduces to
\begin{equation} \label{eq:pde-nomix-cdf}
    \partial_t F(x,t) + \alpha(F(x,t)) x \partial_x F(x,t) = 0.
\end{equation}
Viewed as a non-linear transport equation, this PDE can be solved by the method of characteristics. For this model, it turns out that the quantile processes $Q_p(t)$ are the characteristic curves. 

To see this, consider the parametric curve $(Q(t),t)$ with $Q(0) = q$. Let $P(t)$ represent the value of the cdf along the characteristic and let $P(0) = p_0$ where $p = F(q,0)$.  Together, $Q(t)$ and $P(t)$ must satisfy
\begin{displaymath}
    Q'(t) = \alpha(P(t)) Q(t) \text{ and } P'(t) = 0.    
\end{displaymath}
Given the initial value for $P(t)$, and that is constant, we can rewrite the equation for $Q(t)$ in two ways:
\begin{equation}
    Q'(t) = \alpha(p) Q(t), \,\, Q(0) = q 
\end{equation}
where $p = F_0(q)$. The solution to this ODE is \eqref{eq:curved-rp-nomix-quantiles}. Moreover, since the value of the function is constant along characteristics, Equations \ref{eq:defn-q} and \ref{eq:pde-nomix-cdf} follow. The formula for the density follows from differentiating both sides of \Cref{eq:pde-nomix-cdf} while noting the need to use \Cref{eq:defn-q} when applying the chain rule to $q(x,t)$.

The observation that the characteristics $Q(t)$ are, in fact, quantile curves follows from noting that if $F_0(q) = p$, then $F(Q(t),t) = p$ for all $t > 0$, implying that $Q(t)$ is the p-quantile curve. 
\end{proof}

\subsection{Dynamics of the mean and quantiles}

In the curved-$R(p)$ case, there are no explicit formulas for the evolution of the mean or for the quantiles. Nevertheless, we can derive an ODE that these quantities satisfy. For the purpose of deriving an ODE for the mean, we introduce $\income_p(t)$, which represents the income earned by those in the $P$th percentile or below:
\begin{equation} \label{eq:defn-income-p}
    \income_p(t) \coloneqq \int_0^{Q_p(t)} \!\!\! x f(x,t) \d x.
\end{equation}

\begin{proposition}
\label{prop:curved-rp-meas}
    Let $R(p)$, $f_0$, $F_0$, $f$ and $F$ satisfy the conditions of \Cref{prop:curved-rp-nomix}, but allow $\cv \geq 0$. Then the quantiles $Q_p(t)$, the share of income earned by the lower $100p$-percent satisfy the ODEs
    \begin{equation} \label{eq:ode-quantile}
    \begin{aligned}
        \frac{\d}{\d t}\left(\ln\big(Q_p(t)\big)\right) = \big(\alpha(p) - 2 \beta\big) + \beta Q_p(t) \frac{f'(Q_p(t),t)}{f(Q_p(t),t)}.
    \end{aligned}
    \end{equation}
    where $'$ denotes a partial spatial derivative. Furthermore, 
    \begin{equation} \label{eq:ode-income}
        \ddt \income_p(t) = \int_0^{Q_p(t)} \!\!\! \alpha(F(x,t)) x f(x,t) \d x - \beta Q^2_p(t) f(Q_p(t),t). \hspace{1.1 cm}
    \end{equation}
\end{proposition}

\begin{proof}
For the quantile ODE, we differentiate both sides of the definition of the quantile function with respect to time:
\begin{displaymath}
    \frac{\d}{\d t} \income_p(t) = \frac{\d}{\d t} F(Q_p(t),t).
\end{displaymath}
Taking the total derivative of $F$, writing $f(x,t) = \partial_x F(x,t)$, and denoting the time derivative by '$\cdot$', we have 
\begin{displaymath}
    0 = f(Q_p(t),t) \dot Q_p(t) + \frac{\partial F}{\partial t}(Q_p(t), t)
\end{displaymath}
which, by \eqref{eq:defn-pde-cdf}, yields
\begin{displaymath}
\begin{aligned}
    f(Q_p(t),t) \dot Q_p(t) &= - \alpha(F(Q_p(t),t)) Q_p(t) f(Q_p(t),t) \\
    & \qquad + 2\beta Q_p(t) f(Q_p(t),t) + \beta Q_p^2(t) \frac{\partial f}{\partial x}(Q_p(t), t)
\end{aligned}
\end{displaymath}
Dividing through by $f(Q_p(t))Q_p(t)$ yields \eqref{eq:ode-quantile}.

To derive the ODE for $\income_p(t)$ we introduce the CDF version of \eqref{eq:defn-income-p}, which follows from integration by parts:
\begin{displaymath}
\begin{aligned}
    \income_p(t) &= \int_0^{Q_p(t)} \!\!\! x f(x,t) \d x \\
    &= Q_p(t) F(Q_p(t),t) - \int_0^{Q_p(t)} F(x,t) \d x.
\end{aligned}
\end{displaymath}
Recalling that, by definition, $F(Q_p(t),t) = p$ for all $t$, we differentiate both sides with respect to time to attain two terms of the form $p \dot Q_p(t)$ that cancel, leaving
\begin{equation}
    \ddt \income_p(t) = \int_0^{Q_p(t)} \frac{\partial}{\partial t}F(x,t) \d x.
\end{equation}
Similar to the above, we then apply the PDE definition of the cdf \eqref{eq:defn-pde-cdf}:
\begin{displaymath}
\begin{aligned}
    \ddt I_p(t) &= \int_0^{Q_p(t)} \alpha(F(x,t)) x f(x,t) - 2 \beta x f(x,t) \d x \\ 
    & \qquad - \int_0^{Q_p(t)} \beta x^2 \frac{\partial f}{\partial x}(x,t) \d x 
\end{aligned}    
\end{displaymath}
Integrating the last term by parts causes a cancelation of the integrand term $2 \beta x f(x,t)$ and introduces the term $\beta Q_p^2(t) f(Q_p(t),t)$, which completes equation \ref{eq:ode-income}.
\end{proof}

For the mean, we would like to take the limit of \eqref{eq:ode-income} as $p \to 1$, which is equivalent to taking $\lim_{x \to \infty} x^2 f(x,t)$. But this limit is not finite for distributions with a power law tail that decays too slowly. However, this limit is zero for the lognormal distribution and any other distribution with exponentially decaying tail. 

In the following proposition, we provide an ODE for the mean income assuming that necessary tail behaviors hold for all $t \geq 0$.

\begin{proposition} \label{eq:mean-ode}
    Suppose that a time-dependent income distribution satisfies the conditions of \Cref{prop:curved-rp-meas}. Suppose further that a solution exists and for all $t \geq t_0$,
    \begin{displaymath}
        \lim_{x \to 0} x^2 f(x,t) = 0 = \lim_{x \to \infty} x^2 f(x,t).
    \end{displaymath}
    If $X(t)$ be a stochastic process whose density is $f(\cdot,t)$ for all $t$ then the mean income satisfies
    \begin{equation} \label{eq:ode-mean}
    \ddt E(X(t)) = \int_0^{\infty} \!\!\! \alpha(F(x,t)) x f(x,t) \d x
    \end{equation}
\end{proposition}
\begin{remark}
    Note that while \eqref{eq:ode-mean} appears to not depend on $\beta$, the diffusivity does affect the solution $f(x,t)$ and therefore the mean is modified as a second-order effect.
\end{remark}

\begin{proof}
Differentiating under the integral sign and using the PDE \eqref{eq:defn-pde-pdf} for the pdf, we have
\begin{displaymath}
\begin{aligned}
    \ddt E(X(t)) &= \ddt \int_0^\infty x f(x,t) \d x 
    = \int_0^\infty x \frac{\partial f}{\partial t}(x,t) \d x \\
    &= -\int_0^\infty x \partial_x \big(\alpha(F(x,t),t)) x f(x,t)\big) \d x + \beta \int_0^\infty \!\! x \partial_{xx} \big(x^2 f(x,t)\big) \d x
\end{aligned}
\end{displaymath}
Integrating by parts, the tail and near-zero conditions imply that there are no contributions from the boundaries. It follows that
\begin{displaymath}
\begin{aligned}
    \ddt E(X(t)) = \int_0^\infty x \alpha(F(x,t),t)) f(x,t) \d x + \beta \int_0^\infty \partial_x \big(x^2 f(x,t)\big) \d x
\end{aligned}
\end{displaymath}
For the last term we apply the fundamental theorem of calculus and observe that there is 0 contribution from the $x \to 0$ and $x \to \infty$ limits. The $\beta$ term is eliminated and we have the desired result. 
\end{proof}

\section{Numerical Results}
\label{sec:numerical}

Our ultimate goal is to initiate a conversation about the interplay of inequality and mobility for simple models of evolving income distributions. In this section we first describe methods for fitting our proposed models for the initial (US 1968) distribution and evolution of the income distribution in the US from 1968 to 2021. For the initial distribution we use a log-normal distribution, fit to replicate the mean and Gini coefficient. As described in the introduction, we considered four model regimes defined by two dichotomies: flat versus curved $R(p)$, and $V = 0$ versus $V >0$. For the parameters of the time-evolutions in each case, we used ad hoc methods to best fit the average US GIC values over the 1968-2021 range. The residuals we used are described in \Cref{defn:residuals} and the parameter fitting discussion is recorded in \Cref{sec:numerical:dynamic}.

\begin{figure}[h!]
    \centering    \includegraphics{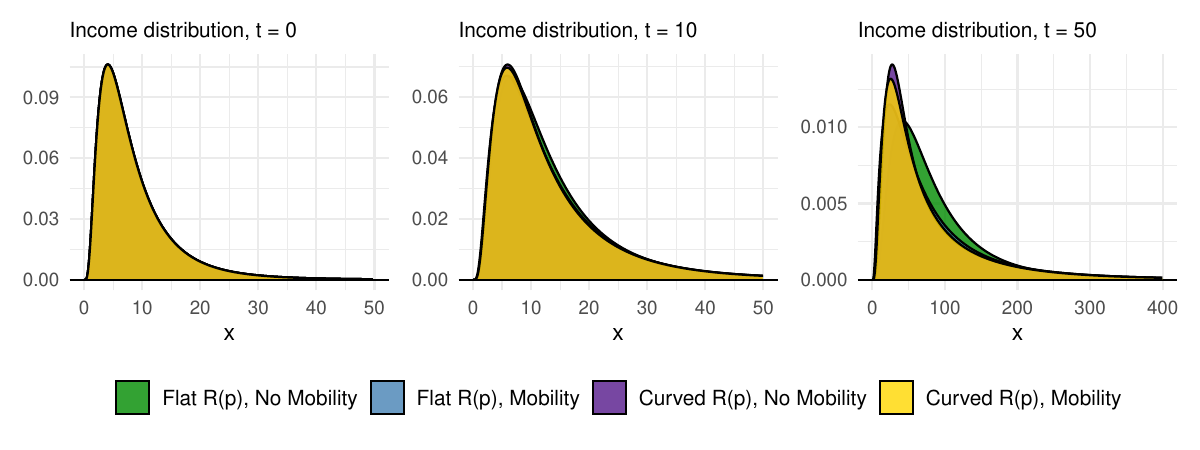}
    \caption{\textsc{Numerical results under the four main models.} An overlaid comparison of pdfs that arise from numerically solving the PDE model with the initial condition described in \Cref{sec:numerical:init} and the $R(p)$ and $\cv$ selections described in \Cref{sec:numerical:dynamic}. The distributions look similar, distinctions in the underlying mobility is laid out in Section \ref{sec:numerical:mobility}. Parameters of the initial distribution and PDE model are given in Tables \ref{tab:lognormal-init} and \ref{tab:parameters-dynamic}.
    \label{fig:numerical-distributions}}
\end{figure}

We simulated all four models using the same engine. We took a log-transform of the original PDE and simulated the system in log-scale (see \Cref{app:numerical:sde}). This had the benefit of allowing for a uniform grid on which to discretize the system. Because the cdf $F$ appears in a classic form for conservation laws with diffusion, we used the cdf to implement our numerical method and then numerically differentiated to report the corresponding pdf, for example in \Cref{fig:numerical-distributions}. To deal with the distinct transport and diffusion mechanisms, we implemented a Strang-splitting method: a half-diffusion step, a whole transport step, and a half-diffusion step within each system time step. At benchmark timesteps we recorded the pdf and cdf in the original income scale and calculated quantities of interest there. To validate the discretization-scheme, we compared the outcome of flat-$R(p)$ models to the exact formulas derived in \Cref{sec:flat-rp}, and compared the outcome of the curved-$R(p)$ model with no randomness to the exact formula given in \Cref{prop:curved-rp-nomix}. 

In terms of model calibration, rather than relying on the Gini as our sole summary statistic for inequality, we made use of the quintile upper limits reported by the US Census. We used the following normalized residual scores, one for individual target distributions and the other for the change among distributions. 

\begin{definition}[Residual Score Functions]
\label{defn:residuals}
    Suppose that there is a target vector of quantile values $\vec{q} = (q_1, q_2, \ldots, q_n)$ that are associated with percentile values $\vec{p} = (p_1, p_2, \ldots, p_n)$. In other words, the income level of the 100$p_i$-percentile is $q_i$. Let $f(x)$ and $F(x)$ be the pdf and cdf of a proposed model distribution for the target values. Then the quantile residual score (QRS) is the squared percentage error defined as follows:
    \begin{equation}
        \qrs = \sum_{i=1}^n \left( \frac{q_i - F^{-1}(p_i)}{q_i} \right)^2.
    \end{equation}
    Similarly, if $\widehat{\gic}(p)$ is the estimated annual growth incident curve for a data set of interest, and $\gic(p)$ is the output of a given model, then the GIC residual score (GICRS) is defined by
    \begin{equation}
        \gicrs = \sum_{i=1}^n \left(\frac{\widehat{\gic}(p_i) - \gic(p_i)}{\widehat{\gic}(p_i)}\right)^2
    \end{equation}
\end{definition}

\subsection{Log-normal fit for initial condition}
\label{sec:numerical:init}

There are a wide set of perspectives on how to best capture an income distribution. The log-normal distribution is still commonly used a baseline ``simplest'' model \citep{bergstrom2022role}. Meanwhile distributions with more parameters (like the Singh-Maddala \citep{maddala1976function,kumar2017singh} or Dagum \citep{dagum1977new,kleiber2008guide} distributions) or different tail behavior (like the Champernowne distribution \citep{champernowne1953model} or laws with a Pareto tail \citep{mandelbrot1960,akhundjanov2020gibrat}) can each be shown to be better fits depending on the income groups \citep{clementi2005pareto,safari2020power} or the metrics being used for the fit \citep{akhundjanov2020gibrat}. Our present purpose is not to optimize the fit, but rather to understand variation in internal dynamics among multiple adequate fits. Moreover, since there are exact formulas for solutions available in the flat-$R(p)$ case with log-normal initial data, the comparisons gained by choosing identical initial conditions outweight benefit gained from better initial fits.

\begin{table}[h!]
    \centering
    \begin{footnotesize}
    \begin{tabular}{c|c|c|c}
       Benchmarks & $\hat \mu$ & $\hat \sigma$ & $\qrs$ \\
       \hline \hline 
       Mean/Median & 
       $\displaystyle \ln(\mathrm{Med}(X))$ & 
       $\displaystyle \sqrt{2 \big(\ln(E(X)) - \ln(\text{Med}(X))\big)}$ & 0.3118 \\
       Mean/Gini & 
       $\displaystyle \ln(E(X)) - \Phi^{-1}\!\!\left(\frac{\big(\text{Gini}(X) + 1\big)}{2} \right)$ & 
       $\displaystyle \sqrt{2} \, \Phi^{-1}\!\!\left(\frac{\big(\text{Gini}(X) + 1\big)}{2}\right)$ & 0.0464\\
       Median/Gini & 
       $\displaystyle \ln(\mathrm{Med}(X))$ & 
       $\displaystyle \sqrt{2} \, \Phi^{-1}\!\!\left(\frac{\big(\text{Gini}(X) + 1\big)}{2}\right)$ & 0.1596
    \end{tabular}
    \end{footnotesize}
    \caption{\footnotesize \textsc{Log-normal parameter fits.} We display the results of using different combinations of the income distribution benchmarks: the mean, median, and Gini coefficient. We report the formulas in terms of a random variable $X$ that has income pdf $f$. As before, $\Phi$ is the cdf of a standard normal random variable. We use the percentile-upper limits from 1968 ($p = (0.2, 0.4, 0.6, 0.8, 0.95)$ and $q = (3.323, 6.300, 9.030, 12.688, 19.850)$, where $\vec{q}$ is expressed in thousands of US dollars) to evaluate each fit in terms of the quantile residual score (QRS), \Cref{defn:residuals}. 
    \label{tab:lognormal-init}}
\end{table}

Interestingly, even when establishing that we want to use log-normal distribution, there is not an unambiguous choice of parameters for best fit. The log-normal distribution has two parameters, so attempting to fit three or more quantities (like the mean, median, and Gini, for example) results in an overdetermined problem. The choice of inferred parameters will depend on what weighting of the distribution properties we prefer. Using the formulas provided in \Cref{table:lognorm}, we derived three fitting schemes based on using different pairings of the mean, median, and Gini. These are displayed in \Cref{tab:lognormal-init} along with an evaluation of the performance of each using the QRS metric.

Perhaps because the Gini is a summary of an entire distribution, fits that use the Gini outperformed those that did not when using a metric based on quantile information. For the purpose of the numerical experiments that follow, we use the Mean/Gini log-normal fit.  In our informal parameter search for the Dagum and Singh-Maddala distributions, we found parameter combinations that yielded lower QRS values, but not so significantly lower that it was worth losing the common basis of comparison to the fully log-normal flat-$R(p)$ baseline model.

\subsection{Selecting the PDE model parameters}
\label{sec:numerical:dynamic}

To conduct our study of the co-evolution of inequality and mobility, we selected four models to focus on: flat-$R(p)$ with and without mobility (i.e., with $V > 0$ and $V = 0$) and curved-$R(p)$ with and without mobility. In each case, our goal was to match the starting and end point properties for US data over the period from 1968 to 2021. Given models that have matching growth in inequality, in the following section we study the mobility properties that emerged in each case. 

Similar to the theme of parameter selection used for the initial distribution, in fitting the dynamic parameters $R(p)$ and $\cv$, we prioritized (1) a philosophy that is rooted in making the best available fit for the fully log-normal model (flat $R(p)$ and log-normal initial condition) and generalizing from there, (2) matching the growth of the mean as the top measureable priority (using \Cref{eq:mean-ode}), so that overall economy growth is the same in all models, and (3) using real-life quantile information to evaluate the quality of the fit. The third priority amounts to running the model and evaluating the GIC of model outputs at the same vector of percentiles $(\vec{p} = (0.2, 0.4, 0.6, 0.8, 0.95)$) and comparing to an estimate of the US GIC over the time period 1968 to 2021. In reporting our methods, we will take $t_0 = 0$ and $X_0 \sim \lognorm{\mu_0}{\sigma_0}$ and we will use the shorthand that $X(t)$ is a stochastic process whose law is given by the outcome of the described model, $X(t) \sim f(\cdot,t)$. 

\begin{figure}
    \centering
    \includegraphics{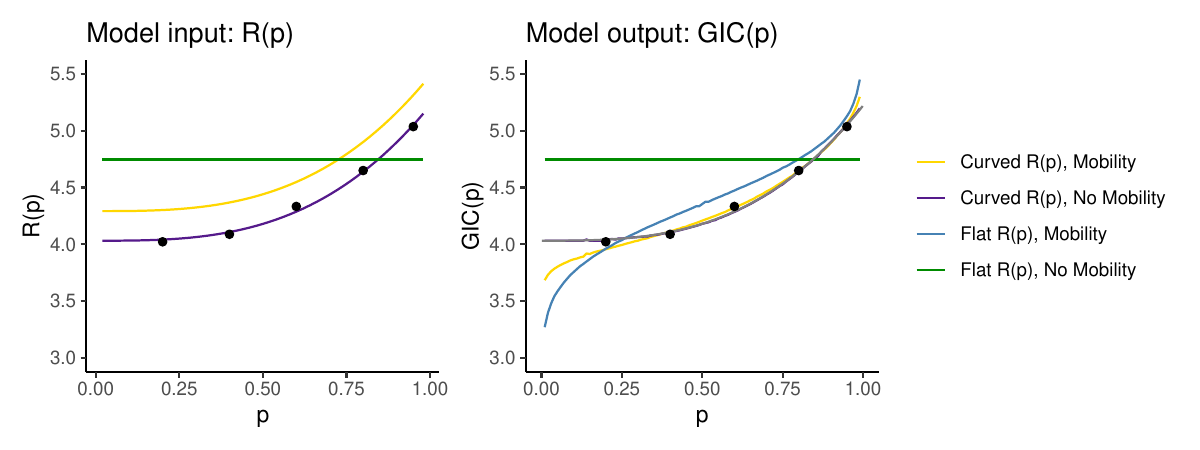}
    \caption{\textsc{Model Calibration, Growth Incidence Curves (GIC).} \textbf{(Left)} The $R(p)$-curves used as input for each of the simulated models. \textbf{(Right)} The GIC computed from simulations using the displayed $R(p)$-curves. For each model, we estimated the $R(p)$-curve that would produce GIC that are as close as possible to US Census data (black dots).}
    \label{fig:gic-results}
\end{figure}

The method for fitting the flat-$R(p)$ models is straightforward. As observed in \Cref{cor:lognormal-mmg}, with $R(p) = R_*$ for all $p \in [0,1]$, the mean of the income distribution satisfies
\begin{displaymath}
    E(X(t)) = \exp\!\left(\mu_0 + \frac{\sigma_0^2}{2} + \alpha t \right).    
\end{displaymath}
We emphasize again that the time-dependent mean does not depend on the randomness parameter $\cv$. Therefore it is natural to fit $\alpha$ accordingly. If $T$ is the final time of the simulation and $M_\mathrm{final}$ is the mean of the income distribution we wish to fit, then 
\begin{equation}
    \hat \alpha = \frac{1}{T} \left(\ln(M_\mathrm{final}) - \Big(\mu_0 +\frac{\sigma_0^2}{2}\Big)\right)
\end{equation}
For ease of comparison to the GIC values reported in the introduction, we write
\begin{equation}
    R_* = 100 \left(e^{\hat \alpha} - 1\right). 
\end{equation}

It remains to find an appropriate value for the randomness parameter. To do this, we used the quantile formula for flat-$R(p)$ reported in \Cref{eq:lognormal-percentile} and sought to minimize the GIC equivalent of the QRS. 

Let $\widehat{\gic}(p)$ be the annual GIC computed from US data over the 1968 to 2021. Let $\gic(p \with R_*, \cv)$ be defined by the formula derived in \Cref{prop:lognormal-gic}. Formally $\cv$ is defined as follows:
\begin{equation}
    \cv = \mathrm{argmin}_{m > 0} \left(\sum_{i = 1}^n \left(\frac{\widehat{\gic}(p_i) - \gic(p_i \with R_*, m)}{\widehat{\gic}(p_i)} \right)^2\right).
\end{equation}
There does not appear to be an explicit formula for this value of $\cv$, so we used a numerical estimate. As reported in \Cref{tab:parameters-dynamic}, when $R_* = 4.49$, the $\gicrs$-minimizing value of $\cv$ is approximately 8.7.

In the curved-$R(p)$ case, there does not appear to be any simple method for fitting the mean. As noted after the presentation of \Cref{prop:curved-rp-meas}, the mean can be expressed as an ODE, written in terms of an integral against the evolving density. However, even when an exact formula for the pdf exists (e.g.~when $\cv=0$, \Cref{eq:curved-rp-soln}), an exact evaluation of the integral in \Cref{eq:ode-mean} is generally not feasible.

However, when $\cv = 0$, we noted in \Cref{prop:curved-rp-nomix} that there is a simple solution for the quantile functions:
\begin{displaymath}
    \cv = 0: \quad Q_p(t) = Q_p(0) e^{\alpha(p) t}. 
\end{displaymath}
For each target percentile $p_i$, the $\cv=0$ model output will generate a $GIC(p)$ that matches $R(p)$. So we can simply use our cubic fit of the annual GIC for US data over the period 1968-2021.

When $\cv > 0$ a distortion occurs and the presence of $\beta$ in the quantile ODE, Equation \Cref{eq:ode-quantile} implies that we have to adjust the $R(p)$ input to achieve the target GIC. To this end, we used the ODE \eqref{eq:ode-quantile} and evaluated the value $\ln(f(Q_p(t),t)$ at the initial time using the pdf of the assumed log-normal initial distribution. For a given value $\cv$, and each desired value $p \in (0,1)$, there is a value of $\beta$ that follows from \Cref{eq:defn-alpha-beta} and then this yields a value $\alpha(p)$ using \Cref{eq:ode-quantile}. This $\alpha(p)$ can then be translated into an approximately cubic $R(p)$, whose parameters are given in \Cref{tab:parameters-dynamic}. We ran the full simulation for a range of randomness parameter values and found the lowest residual occurred with $\cv = 5.0$. The parameter selections are summarized in \Cref{tab:parameters-dynamic}.

\subsection{The same growth in inequality can coexist with different mobility stories.}
\label{sec:numerical:mobility}

In \Cref{fig:mmg}, we display the evolution of the mean, median, and Gini for simulated data under our four model assumptions. The fits are not exact, especially under the curved-$R(p)$ models, because there is no explicit formula for the mean over time and because the $GIC$ was used for calibration rather than the median or Gini. We emphasize in particular that the Gini curves displayed in \Cref{fig:mmg} are emergent properties of the model.

\begin{figure}
    \centering
    \includegraphics[width=0.9\textwidth]{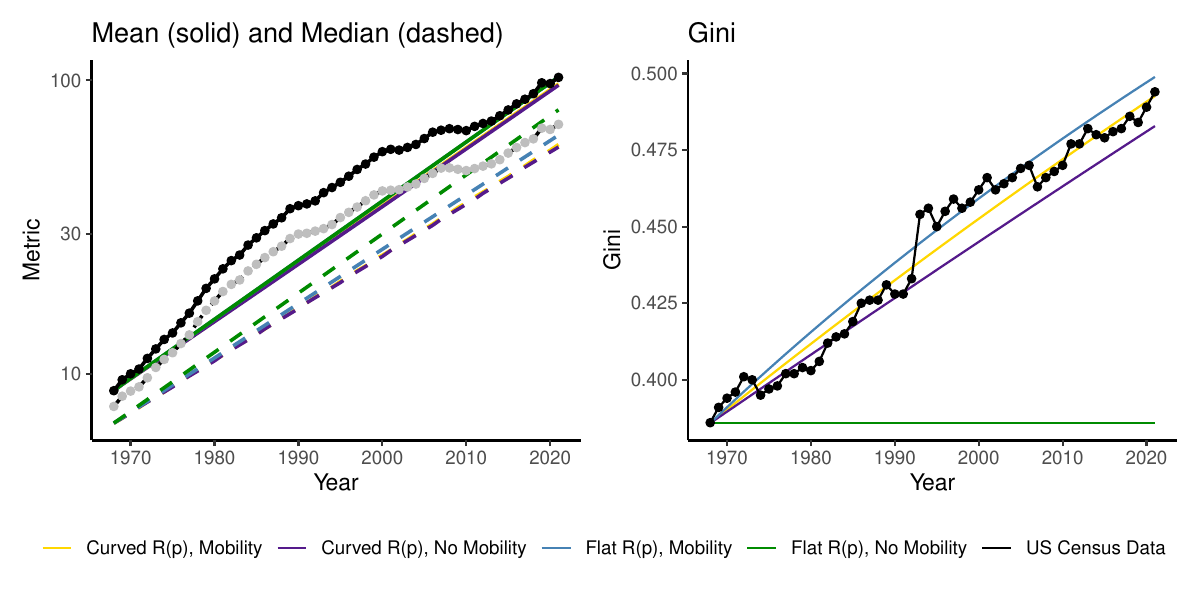}
    \caption{\textsc{Model Calibration: Mean, Median, Gini.} The mean (left panel, solid curves), median (left panel, dashed curves), and Gini (right panel, solid curves) were computed for each model simulation and compared to US census data (dash-dot curves). Mean income was used at the primary calibration metric. When possible, the Gini coefficient was the second consideration. 
    }
    \label{fig:mmg}
\end{figure}

To make the comparison among models more quantitative, in \Cref{tab:results}, we display the calculated values for the mean, median, Gini, and GICRS for each model. Moreover, in the first few rows we report these values for the initial distribution fit (US 1968) and a mean/Gini log-normal fit for the final year (US 2021). 

As would be expected, the flat-$R(p)$ model (egalitarian among income cohorts) with no randomness (egalitarian within cohorts) performs substantially worse than the other models. This is because it contains neither of the mechanisms that we have introduced to produce growing inequality (larger annual growth for higher percentiles and randomness). The other three models produced 2021 results that have residuals on a scale similar to the best log-normal fit. Interestingly, the flat-$R(p)$ model with randomness performed the best. This is most likely a reflection of the fact that our parameterization methods for curved-$R(p)$ models is still immature. A more profound understanding of how curved-$R(p)$ affects the distribution over time will be required to improve the approximation we made based on the initial data. A full accounting of the results appears in \Cref{tab:results}.

\begin{table}[h]
    \centering
    \begin{footnotesize}
    \begin{tabular}{c|c|c|c|c|c}
         & Year & Mean & Median & Gini & GICRS \\
         \hline \hline
        US Census & 1968 & \$8,760 & \$7,743 & 0.386 & --  \\
        Log-normal fit & 1968 & \$8,760 & \$6,792 & 0.386 & 0.0464\\
        \hline US Census & 2021 & \$102,316 &\$70,784 & 0.494 & --\\
        Log-normal fit & 2021 & \$102,796 & \$66,023 & 0.494 & 0.0193  \\
        \hline Flat $R(p)$, no mobility & 2021 & \$102,683 & \$79,335 & 0.385 & 0.3519 \\
        Flat $R(p)$, with mobility & 2021 & \$102,683 & \$65,025 & 0.499 & 0.0214 \\
        Curved $R(p)$, no mobility & 2021 & \$96,280 & \$59,260 & 0.483 & 0.0615 \\
        Curved $R(p)$, with mobility & 2021 & \$97,258 & \$60,422 & 0.493 & 0.0494
    \end{tabular}
    \end{footnotesize}
    \caption{\footnotesize \textsc{Model Results.} Results under our four basic model assumptions, using the mean and Gini from 1968 to parameterize the initial condition. The distribution quality score (GICRS) is given in \Cref{defn:residuals}.
    \label{tab:results}}
\end{table}

The key takeaway from \Cref{fig:mmg} and \Cref{tab:results} is that multiple configurations of $R(p)$ and $\cv$ achieve comparable endpoints in terms of the shape of the income distribution. Within those populations, we will see that the ability of agents to move from one percentile to another can change dramatically. In \Cref{fig:mobility-full}, we display the results of a mobility study in which we show the distribution of three different cohorts after fifty years have evolved according to the respective models. The red distribution is the endpoint of those who started in the 20th percentile in the initial year (1968 date), and the orange and yellow correspond to the initial 50th and 80th percentiles. 

Let $F(x,t)$ denote the CDF of a given model after $t$ years. The subgroup distributions seen in \Cref{fig:mobility-full} correspond to the law of the following agent-based model. Let $p \in (0,1)$ be given and let $Q_p(t)$ be the time-dependent quantile function for our model of interest. Then we define the initial-$100p$-percentile cohort by the SDE
\begin{equation}
\begin{aligned}
    \d \widetilde{X}_p(t) &= \alpha\big(F(\widetilde{X}_p(t),t)\big) \widetilde{X}_p(t) \d t + \sqrt{2 \beta} \widetilde{X}_p(t) \d \widetilde{W}(t), \\
    \widetilde{X}_p(0) &= Q_p(0).
\end{aligned}
\end{equation}
Note that the infinitesimal mean and variance are governed by the full population, but the initial data is a subset of the initial full distribution. Using standard SDE theory, the law of $\widetilde{X}_p(t)$ then satisfies
\begin{equation}
    \partial_t \tilde{f}_p(x,t) + \partial_x\big(\alpha(F(x,t)) x \tilde{f}_p(x,t)\big) = \beta \partial_{xx} (x^2  \tilde{f}_p(x,t)).
\end{equation}
Again, we emphasize that the $F(x,t)$ that serves as the argument for $\alpha$ is the cdf of the full distribution, rather than the cdf of the smaller cohort. After taking a log-transform, this becomes a linear transport equation with diffusion and numerical implementation is straightforward.
\begin{figure}[t!]
    \centering
    \includegraphics[width = 0.8\textwidth]{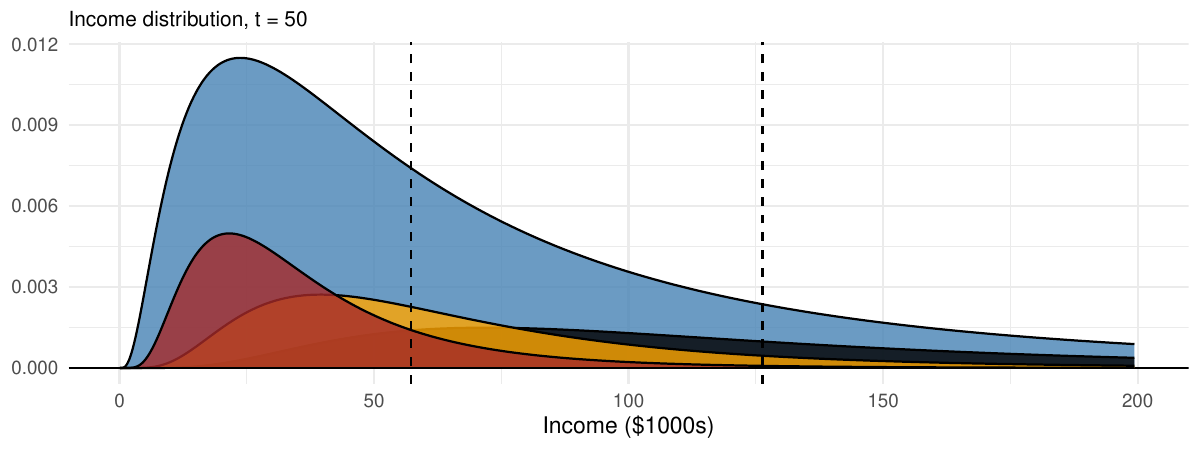}\\
    \includegraphics[width = 0.8\textwidth]{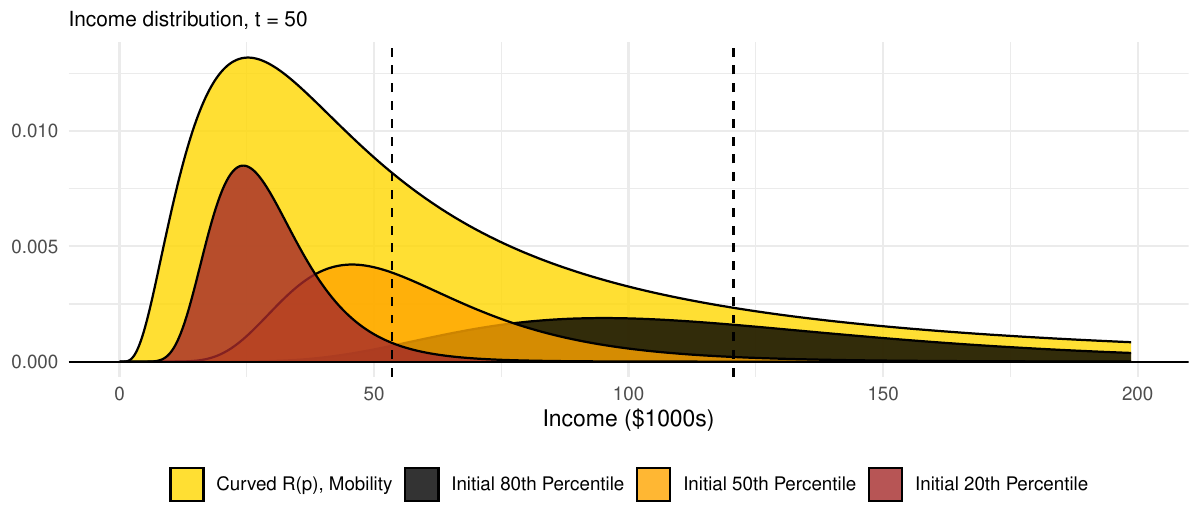}
    \caption{\textsc{Mobility, visualized.} For the Flat-$R(p)$-with-Randomness \textbf{(Top)} and Curved-$R(p)$-with-Randomness \textbf{(Bottom)} models, we display the 50th-year distributions of the year-zero 20th (brown), 50th (orange), and 80th (black) percentile cohorts. See also \Cref{tab:mobility}.
    \label{fig:mobility-full}}
\end{figure}

In the visualization, \Cref{fig:mobility-full}, note that the 20th- and 50th-percentile cohorts in the flat-$R(p)$ model are qualitatively more spread out after 50 years than those same cohorts from the curved-$R(p)$ model. In particular, the area in the red distribution that is to the right of the vertically-indicated final-year 50th-percentile income (dashed line) is substantially larger than the red distribution in the curved-$R(p)$ model. This indicates that there is less mobility from the 20th to 50th percentiles in the curved-$R(p)$ model over time. We calculated various $(P_1 \to P_2)$-mobilities and report them in \Cref{tab:mobility}. Because the average growth rate of the upper percentiles is so much larger than that of the lower percentiles, the mobility that emerges from randommess is greatly diminished.

\begin{table}[h!]
    \centering
    \begin{tabular}{r|c|c|c|}
        & \multicolumn{3}{c|}{\textbf{Fifty-year mobility type}} \\
        & $20 \to 50$ & $50 \to 80$ & $20 \to 80$ \\
        \hline \hline 
        Curved $R(p)$ with randomness & 3.54\% & 1.92\% & 0.005\% \\
        Flat $R(p)$ with randomness & 16.05\% & 9.78\% & 1.13\%
    \end{tabular}
    \caption{Mobility results for the given models. The 50-year $(P_1 \to P_2)$-Mobility is the probability that a member of the $P_1$-percentile at time zero has an income level above the $P_2$-percentile at a final time $T$. Results are from using the parameter choices listed in \Cref{tab:parameters-dynamic}.}
    \label{tab:mobility}
\end{table}

\section{Discussion}

We have conducted a theoretical investigation of the dual dynamics of income inequality and economic mobility through a pairing of mathematical models that operate at the population scale and individual scale. We have shown that multiple assumptions about the mechanisms of income distribution evolution can produce similar inequality growth while telling very different mobility stories. When we assume that the mean annual percentage income growth of individuals is larger for individuals in higher percentile ranks, the model produces results that echo the findings of \citet{akee2019race} that there is more local mobility in the lower income ranks than in the upper ranks, but global mobility is decreasing, with the top of the distribution moving ever faster away. 

The implication for decision-making under Rawls's ``Veil of Ignorance'' is compelling. The presence of upward mobility may be enticing when considering what levels of inequality are tolerable (with implications for what redistribution policies are demanded), but the mobility profile is constantly shifting under foot. Income inequality is not a static object, and our ability to traverse the ranks is rife with misperception and skewed toward unrealistic optimism. 

From a policy perspective, imposing egalitarianism among the percentile ranks (flattening the $R(p)$-curve) leads to greater mobility, but we find it is still not enough. Through year-to-year variability alone, inequality will continue to grow. And while the trope that ``everything mixes in the long run'' is technically true, we show in \Cref{sec:analysis:flat-rp:mobility} that the ``long run'' gets longer and longer over time. 

While our analysis identifies fundamental interactions between inequality and mobility, the framework invites investigation of critical relationships between model parameters. For example, we see that mobility is greatly challenged when the percentile-dependent growth curve is increasing, but is there a degree of curve that makes traversing the percentile ranks mathematically impossible? In the present analysis we only considered non-decreasing $R(p)$-curves, which are the predominant form seen in developed economies, and so our studies resulted only in scenarios in which inequality is increasing over time. It stands to reason that a decreasing $R(p)$-curve can yield decreasing inequality, but what is the ratio of growth between low and high income individuals that is necessary to overcome the inequality growth we have found to be intrinsic due to randomness?

The model we have introduced is deliberately minimal in its assumptions.  We did this to emphasize that surprises exist even in the simplest of settings, but this also compels a deliberate development of the inequality/mobility constrast in the presence of more complex and realistic assumptions. For example, we make no effort to incorporate certain population realities first considered in the context of stochastic processes by \citet{champernowne1953model}. To that end, we did not consider the impact of retirement and the entrance of new, younger participants into the income distribution. Establishing a link between individuals exiting and entering the work force would provide a means to model and study intergenerational effects recently highlighted by the Great Gatsy Curve introduced by \citet{krueger2012Gatsby}, \citet{corak2013income} and \citet{carroll2016income}. Champerknowne also considered career-dependence in individuals' income growth rates -- an example of which could be a PhD student and a service-industry worker may have the same present-day salary, but different expected income growth over decades. This kind of stratified growth (whether due to career or age) can lead to heavy-tailed distributions \citep{champernowne1953model,reed2003pareto}. But we note that most derivations of power-law tails from stochastic process models rely on studying a stationary distribution that is assumed to exist (see also, \cite{chatterjee2007econophysics}). Our position is that the steadily-growing Gini coefficient is direct evidence that we are not in a stationary regime, and in fact, the emphasis of study should be on non-equilibrium dynamics. 

Having said this, the model we present here can be seen as a framework for revisiting the various fascinating decompositions of mobility that exist in the literature, from \citet{creedy2002income} to the the study of poverty reduction introduced by \citet{bergstrom2022role}.  Our conclusion that mobility is greatly affected by whether growth is egalitarian among cohorts (the flat- vs.~curved-$R(p)$ dichotomy) highlights the need to study an evolving income distribution's internal mixing, with profound implications for deciding how pro-low-income a given growth profile really is \citep{son2007pro}.

\bibliographystyle{plainnat}
\bibliography{ineq.bib}

\appendix

\section{Numerical methods}
\label{app:numerical}

There are two fundamental challenges to overcome in simulating this model. The first is the exponential growth of the income values, and the second is the state-dependence of the drift term. The exponential growth of values can be overcome by conducting an appropriate log-transform which then allow for simulation on a fixed set of mesh points. To deal with the state-dependence of the drift term, we sought representations of the dynamics that can be written in conservation form. For the main system, this meant simulating the cdf of the log-transformed system. But for the evolution of subsets of the data, it is better to work with the pdf of the log-scale system, given the values of the log-scale cdf for the full population.

\subsection{Log-scale SDE}
\label{app:numerical:sde}

The asymptotic growth of geometric Brownian motion (assuming constant $\alpha$ and $\beta > 0$) is best understood by looking at the dynamics of its logarithm. If $X(t)$ satisfies \eqref{eq:defn-sde-X} with constant $\alpha$, then a quick application of It\^o's formula shows that $\ln(X(t))$ satisfies
\begin{displaymath}
    \d [\ln(X(t))] = (\alpha - \beta) \d t + \sqrt{2 \beta} \d W(t),
\end{displaymath}
where $W(t)$ is a standard Brownian motion.
This is the logarithmic representation of the dynamics, introduced by \citet{fernholz2002stochastic} and developed by \citet{fernholz2009stochastic}. 
In integral form, this can be written 
\begin{equation} \label{eq:defn-Y}
   \ln(X(t)) = \ln(X(0)) + (\alpha - \beta) t + \sqrt{2 \beta} W(t).
\end{equation}

More generally, suppose that we write the growth term in the form
\begin{equation} \label{eq:defn-alpha-gamma}
    \alpha(q) = a + \gamma(q),
\end{equation}
where $a = \alpha(0)$. Moreover, we can assign a location parameter $\mu_0$ and scale parameter $\sigma_0$ for the initial condition. This will be useful when we analyze the log-normal model below, and also provides a uniform scale when simulating the log-scale dynamics starting from different initial probability models. With these parameters in mind, define 
\begin{equation}
    Y(t) \coloneqq \frac{1}{\sigma_0} \left(\ln(X(t)) - (\mu_0 + (a - \beta) t)\right).
\end{equation}
For every $t \geq 0$, consider the transformation $\phi_t$ and its inverse $\psi_t$
\begin{equation}
\begin{aligned}
    y &= \vf_t(x) = \big(\ln(x) - (\mu_0 + (a - \beta) t)\big)/\sigma_0, \\
    x &= \psi_t(y) = e^{\sigma_0 y + \mu_0 + (a - \beta) t}.
\end{aligned}
\end{equation}
(Note that these functions are different than those introduced in Section \ref{sec:curved-rp}.) Using the change of variable formula, the density of CDF of $Y(t)$ is related to that of $X(t)$ through
\begin{equation} \label{eq:distr-transform}
    h(y,t) = \sigma_0 x \f(x,t), \qquad H(y,t) = F(x,t).
\end{equation}

Using It\^{o}'s formula, we can derive an SDE for $Y(t)$, namely 
\begin{equation} \label{eq:log-sde}
    d Y(t) = \frac{1}{\sigma_0} \gamma(H(Y(t),t)) \d t + \frac{\sqrt{2 \beta}}{\sigma_0} \d W(t).
\end{equation}
This follows from writing $\f(y,t) = f_t(y)$ for clarity, and observing that
\begin{equation}
\begin{aligned}
     d \vf(X(t), t) &= \frac{\partial \vf}{\partial t}(X(t), t) \d t + \frac{\partial \vf}{\partial x}(X(t), t) \d X(t) \\
    &\qquad \qquad + \frac{1}{2} \frac{\partial^2 \vf}{\partial^2 x}(X(t), t) \langle \d X(t), \d X(t) \rangle \\
    &= - \frac{a - \beta}{\sigma_0} \d t + \frac{1}{\sigma_0}\left((a + \gamma(\F(X(t),t)) \d t + \sqrt{2 \beta} \d W(t)\right) + \frac{1}{2} \frac{2 \beta}{\sigma_0} \d t 
\end{aligned}
\end{equation}
which then simplifies to the claimed SDE.

\subsection{Log-scale PDE}
\label{app:numerical:H}

From \eqref{eq:distr-transform} and \eqref{eq:log-sde}, we have that the log-transformed income distribution $h(y,t)$ satisfies
\begin{equation} \label{eq:log-pde}
\begin{aligned}
    \partial_t h(y,t) + \frac{1}{\sigma_0} \partial_y (\gamma(H(y,t)) h(y,t)) &= \frac{\beta}{\sigma_0^2} \partial_{yy} \, h(y,t), \\
    h(y,0) &= \sigma_0 x \f_0(x).
\end{aligned}
\end{equation}
Integrating both sides with respect to $y$, and noting that $H(y,t) = \int_0^y h(\upsilon,t) \d \upsilon$, we have
\begin{equation}
\partial_t H(y,t) + \frac{1}{\sigma_0} \gamma(H(y,t)) h(y,t) = \frac{\beta}{\sigma_0^2} \partial_{yy} H(y,t).
\end{equation}
This follows under the assumptions that that $h(0,t) = 0$ (we are not tracking individuals with zero income) and $\partial_y h(0,t) = 0$ (there is not a flux of individuals across the zero-income level.

Finally, we observe that when $\alpha(q)$ has the proposed form $\alpha(q) = a + \delta q^\nu$, we can view $H$ as satisfying a porous media equation. This connection was first observed by \citet{jourdain2013propagation} in their work generalizing the Atlas model to a wider-range of rank-based models:
\begin{equation} \label{eq:log-pde-nu}
    \partial_t H(y,t) + C\partial_y(H(y,t)^{\nu + 1})  = D \partial_{yy} H(y,t)
\end{equation}
where $C = \delta / (\nu + 1) \sigma_0$ and $D = \beta / \sigma_0^2$. This is precisely the form for which Jourdain established existence and uniqueness \citep{jourdain1997diffusions}.

\end{document}